\documentclass[a4paper,11pt,reqno]{amsart}
\usepackage{amsmath, amsthm, amssymb, mathrsfs, amsfonts}
\usepackage{array}
\usepackage[active]{srcltx}
\usepackage{graphicx}
\usepackage[margin=3cm]{geometry}
\usepackage[english]{babel}
\usepackage[colorlinks=true, allcolors=blue]{hyperref}
\usepackage{enumerate}
\usepackage{float}

\usepackage{enumerate}
\usepackage{enumitem}

\theoremstyle{plain}
\newtheorem{theorem}{Theorem}[section]

\newtheorem{proposition}[theorem]{Proposition}
\newtheorem{corollary}[theorem]{Corollary}

\theoremstyle{definition}
\newtheorem{definition}[theorem]{Definition}
\newtheorem{remark}[theorem]{Remark}
\newtheorem{example}[theorem]{Example}

\def\gr{\mathop{ \rm gr}\nolimits}

\def\diam{\mathop{ \rm diam}\nolimits}

\DeclareMathOperator{\wt}{wt}
\DeclareMathOperator{\supp}{supp}
\DeclareMathOperator{\Syn}{S}

\usepackage{tikz,xcolor,hyperref}

\definecolor{lime}{HTML}{A6CE39}
\DeclareRobustCommand{\orcidicon}{%
	\begin{tikzpicture}
	\draw[lime, fill=lime] (0,0)
	circle [radius=0.16]
	node[white] {{\fontfamily{qag}\selectfont \tiny ID}};
	\draw[white, fill=white] (-0.0625,0.095)
	circle [radius=0.007];
	\end{tikzpicture}
	\hspace{-2mm}
}

\foreach \x in {A, ..., Z}{%
	\expandafter\xdef\csname orcid\x\endcsname{\noexpand\href{https://orcid.org/\csname orcidauthor\x\endcsname}{\noexpand\orcidicon}}
}

\title{Operations on binary linear codes and their associated graphs}

\author[J.H. Castillo]{John H. Castillo\orcidB{}}
\address{John H. Castillo, Departamento de Matem\'aticas y Estad\'istica, Universidad de Nari\~no}
\email{jhcastillo@udenar.edu.co}

\author[L. Delgado]{Lisbeth D. Delgado-Ordoñez\orcidC{}}
\address{Lisbeth D. Delgado-Ordoñez, Departamento de Matem\'aticas, Universidad del Cauca}
\email{lddelgado@unicauca.edu.co}

\author[A. Holgu\'in-Villa]{Alexander Holgu\'in-Villa\orcidA{}}
\address{Alexander Holgu\'in-Villa, Escuela de Matem\'aticas, Universidad Industrial de Santander}
\email{aholguin@uis.edu.co}

\keywords{Coset leader, code operations, graph operations, Hasse diagram, linear code.}
\subjclass[2020]{05C76, 05C60, 05C62, 94A24}

\usepackage{subcaption}
\graphicspath{../FIGURAS/}

\begin{document}
 \noindent
  \renewcommand{\refname}{References}


\begin{abstract}

This manuscript addresses a connection between operations in coding theory and graph theory, as well as some operations in both areas. The concept of the Hasse diagram of a binary linear code, $\Gamma(C)$, is recall; this concept visually represents the cosets of a binary linear code with a partial ordering. The main objetive is to give a characterization of the graphs corresponding to: extended codes, punctured codes, and the direct sum of two codes, using the graph of the given binary linear code as a starting point. Relationships are established between the graphs of code operations and particular operations of graphs.
\end{abstract}

\maketitle

\section{Introduction}
Graphs are mathematical objects which consist of points (vertices) and  lines (edges) connecting none, some, or all possible
pairs of vertices and its study, graph theory, has been widely applied in various branches of mathematics, including coding theory, where the
graph theory plays a fundamental role by providing structural frameworks to design, analyze, and decode error-correcting codes and therefore,
the connection between them has been extensively studied by several researchers over the past fifty years; see for instance \cite{banihashemi2001tanner,etzion1999codes,halford2006codes,jungnickel1996codes,kschischang2003codes,mallik2021graph,rouayheb2012graph,tanner1981recursive,tonchev2002error}.

Let $\mathbb{F}_2$ be the field with $2$ elements. A \emph{binary $[n,k]$-linear code} is a $k$-dimensional subspace of $\mathbb{F}_2^n$.
An element of a binary linear code is called a \emph{codeword}. The \emph{Hamming distance}, $d(\boldsymbol{x},\boldsymbol{y})$, between two codewords $\boldsymbol{x}=(x_1,\ldots,x_n), \boldsymbol{y}=(y_1,\ldots,y_n)\in \mathcal{C}\subseteq\mathbb{F}_2^n$ is the number of entries where they differ; i.e, $d(\boldsymbol{x},\boldsymbol{y})=d(\boldsymbol{x},\boldsymbol{y}) =|\{i:x_i\neq y_i, ~1\leq i\leq n\}|$.

For $\boldsymbol{x}\in\mathbb{F}_2^n$, the \emph{Hamming weight} of $\boldsymbol{x}$ is $\wt(\boldsymbol{x})=d(\boldsymbol{x},\boldsymbol{0})$. The \emph{minimum distance} $d(\mathcal{C})=d$ of a linear
code $\mathcal{C}$ is defined as the minimum weight among all non-zero codewords, thus we called it a binary $[n,k,d]$-linear
code. A \emph{generator matrix} for an $[n, k]$-linear code $\mathcal{C}$ is any $k\times n$ matrix $G$ whose rows form a basis for
$\mathcal{C}$. So the code $\mathcal{C}$ can be seen as

\begin{equation}\label{generator_matrix}
\mathcal{C}=\{ \boldsymbol{x}G: \boldsymbol{x}\in \mathbb{F}_2^k\}.
\end{equation}

Also, as a binary linear code is a subspace of a vector space, there is an $(n-k)\times n$ matrix $H$, called a \emph{parity-check matrix} for the $[n, k]$-linear code $\mathcal{C}$, such that

\begin{equation}\label{check_matrix}
\mathcal{C} =\{\boldsymbol{x} \in \mathbb{F}_2^n:  H\boldsymbol{x}^T = \boldsymbol{0}\}.
\end{equation}

Since a code $\mathcal{C}$ is an elementary abelian subgroup of the additive group $\mathbb{F}_2^n$, it can be partitioned in cosets $\boldsymbol{x} + \mathcal{C}=\{\boldsymbol{x}+\boldsymbol{c}:\boldsymbol{c}\in \mathcal{C}\}$. Two vectors $\boldsymbol{x}$ and $\boldsymbol{y}$ belong to the same coset if and only if $\boldsymbol{y}-\boldsymbol{x} \in \mathcal{C}$. We denote with
$\mathfrak{cl}(\mathcal{C})=\{\boldsymbol{x} + \mathcal{C}: \boldsymbol{x}\in \mathbb{F}_2^n\}$ the set of cosets of the code
$\mathcal{C}$.
Let $H$ be a parity-check matrix for $\mathcal{C}$, that is a generator matrix for the set $\mathcal{C}^{\perp}=\{\boldsymbol{x}\in \mathbb{F}_2^n: \boldsymbol{x}\cdot \boldsymbol{c}=0, \text{ for all $\boldsymbol{c}\in \mathcal{C}$}\}$. The \emph{syndrome} of a vector $\boldsymbol{y} \in \mathbb{F}_2^n$ with respect to $H$ is the vector $S(\boldsymbol{y}) = H\boldsymbol{y}^T \in \mathbb{F}_2^{n-k}$. By \eqref{check_matrix} the syndrome of a codeword is $\boldsymbol{0}$ and in general it can be proved that two vector have the same syndrome if and only belongs to the same coset, see \cite[Theorem 1.11.5]{Huffman2003}. Thus there is a one-to-one correspondence between cosets of
$\mathcal{C}$ and syndromes. A vector in a coset with the smallest weight is called a \emph{coset leader}. The zero vector is the unique coset leader of the code $\mathcal{C}$. More generally, every coset of weight at most $t = \lfloor(d-1)/2\rfloor$ has a unique coset leader.

There is a natural partial ordering $\preceq$ on the vectors in $\mathbb{F}_2^n$, which is defined as follows: for
$\boldsymbol{x},\boldsymbol{y}\in \mathbb{F}_2^n$, $\boldsymbol{x}\preceq \boldsymbol{y}$ provided that $\supp(\boldsymbol{x}) \subseteq \supp(\boldsymbol{y})$,
where $\supp(\boldsymbol{c})$ for $\boldsymbol{c}\in \mathbb{F}_2^n$ denote the {\it support} of the vector $\boldsymbol{c}$, i.e.,
the set of coordinates where $\boldsymbol{c}$ is non-zero. If $\boldsymbol{x}\preceq \boldsymbol{y}$, we will also say that
$\boldsymbol{y}$ covers $\boldsymbol{x}$. We now use this partial order on $\mathbb{F}_2^n$ to define a partial order, also
denoted $\preceq$, on the set of cosets of a binary linear code $\mathcal{C}$ of length $n$. If $\mathcal{C}_1$ and
$\mathcal{C}_2$ are two cosets of $\mathcal{C}$, then $\mathcal{C}_1 \preceq \mathcal{C}_2$ provided there are coset
leaders $\boldsymbol{x}_1$ of $\mathcal{C}_1$ and $\boldsymbol{x}_2$ of $\mathcal{C}_2$ such that $\boldsymbol{x}_1 \preceq \boldsymbol{x}_2$.
As usual, $\mathcal{C}_1 \prec\mathcal{C}_2$ means that $\mathcal{C}_1 \preceq \mathcal{C}_2$ but $\mathcal{C}_1\neq \mathcal{C}_2$.
Under this partial ordering the set of cosets of $\mathcal{C}$ has a unique minimal element, the code $\mathcal{C}$
itself. Since, the set $\mathfrak{cl}(\mathcal{C})$ is a  partially ordered set, \emph{poset} for short, with respect to
$\preceq$, we can use a Hasse diagram to represent the partial order on $\mathfrak{cl}(\mathcal{C})$.

A \emph{graph} $\Gamma$ with $n$ vertices is a pair $(V,E)$ where $V=\{v_1, v_2,\ldots ,v_n\}$ is the set of vertices and
$E \subseteq V\times V$ is the set of edges. Given two vertices $v_i$ and $v_j$, if $v_iv_j\in E$, then $v_i$ and
$v_j$ are said to be \emph{adjacent} or that $v_i$ and $v_j$ are \emph{neighbors}. In this case, $v_i$ and $v_j$ are said to be
the end vertices of the edge $v_iv_j$. If $v_iv_j\notin E$, then $v_i$ and $v_j$ are \emph{non-adjacent}. Furthermore, if
an edge $e$ has a vertex $v_i$ as an end vertex, we say that $v_i$ is incident with $e$. The \emph{order} of $\Gamma$ is the cardinal of the set of vertices and the \emph{size} of $\Gamma$ is its number of edges.

In \cite{Delgado2024}, the authors introduced a graphical representation for binary linear codes, as follows.
\begin{definition}
Given a binary linear code $\mathcal{C}$, we denote by  $\Gamma(\mathcal{C})=(V_{\mathcal{C}},E_{\mathcal{C}})$ the  graph
constructed on this wise: the set of vertices  $V_{\mathcal{C}}=\mathfrak{cl}(\mathcal{C})$  and $\mathcal{C}_1\mathcal{C}_2\in E_{\mathcal{C}}$
if $\mathcal{C}_1\prec \mathcal{C}_2$ and $\wt(\mathcal{C}_1)=\wt(\mathcal{C}_2)-1$. If $\mathcal{C}_1\mathcal{C}_2$ is an edge of $\Gamma(\mathcal{C})$, $\mathcal{C}_1$ is called a \emph{child} of $\mathcal{C}_2$, and $\mathcal{C}_2$ is a \emph{parent} of $\mathcal{C}_1$. Actually, the graph $\Gamma(\mathcal{C})$ is known as the \emph{Hasse diagram} of the poset $(\mathfrak{cl}(\mathcal{C}),\preceq)$.  It is clear that $|V_{\mathcal{C}}|=2^{n-k}$.
\end{definition}

\begin{remark} If $\mathcal{C}_1$ and $\mathcal{C}_2$ are cosets of $\mathcal{C}$ with $\mathcal{C}_1 \prec \mathcal{C}_2$, then $\mathcal{C}_1$ is called a \emph{descendant} of $\mathcal{C}_2$, and
$\mathcal{C}_2$ is an \emph{ancestor} of $\mathcal{C}_1$. Note that the coset $\mathcal{C}_0=\boldsymbol{0}+\mathcal{C}$ is always a descendant of any coset of the linear code $\mathcal{C}$, in other words $\mathcal{C}_0$ is minimal in the poset $(\mathfrak{cl}(\mathcal{C}),\preceq)$. Always, this vertex will be represented with a red dot in the figures. A coset of $\mathcal{C}$ is called an \emph{orphan} whenever it has no parents, that is, it is a maximal element in the poset $(\mathfrak{cl}(\mathcal{C}),\preceq)$. These vertices will be denoted, in the figures, with a blue dot, see Figure \ref{fig:ext2}.
\end{remark} 
The main results about the graph $\Gamma(\mathcal{C})$ obtained \cite{Delgado2024} are summarized in the next result. 
\begin{theorem}\label{maintheoremarticle1}
Let $\mathcal{C}$ be an $[n,k]$-binary linear code. Then
\begin{enumerate}[label=\emph{(\arabic*)}]
\item \label{thmmain1}$|V_{\mathcal{C}}|=2^{n-k}$.
\item $\Gamma(\mathcal{C})$ is connected. 
\item $\Gamma(\mathcal{C})$ is a free-triangle graph.
\item $\Gamma(\mathcal{C})$ is bipartite graph. 
\item If $\Gamma(\mathcal{C})$ contains a cycle, then $\gr(\Gamma(\mathcal{C}))=4$.
\item $\diam(\Gamma(\mathcal{C}))\geq \rho(\mathcal{C})$.
\item $\Gamma(\mathcal{C})\cong K_{1,2^{n-k}-1}$ is an star graph if and only if $\rho(\mathcal{C})=1$.
\end{enumerate}
\end{theorem}

In this paper, we study three operations of binary linear codes: puncturing, extending and direct sum. The main interest of this manuscript is to establish the connections between the Hasse diagram of the original code(s) and the graphical representation of the new code.

\section{Coding operations}\label{sec3}
In the classical coding theory, it is a well-known fact that there exist several techniques to obtain new codes from old, including puncturing,
extending and direct sum. In this section, we study these three operations on binary linear codes and examine their effect on the graphical representation of the newly obtained codes. The main objective is to investigate the relationship between the graph of the original code and the graph of the resulting new code. 
\subsection{Puncturing a code}
 \begin{definition}[Puncturing a code]
Let $\mathcal{C}$ be an $[n,k,d]$-binary linear code, \emph{puncturing $\mathcal{C}$} is the operation in which one or more coordinate positions are removed from the codewords. For $1\leq i\leq n$, we denote with $\mathcal{C}^{*_i}$, puncturing the code $\mathcal{C}$ in the $i$-th coordinate. 

\end{definition}

\begin{example}\quad \label{ejemploperforado}
\begin{enumerate}[label=(\alph*)]
\item \label{ejemploper1}
  Consider the $[3,1,3]$-binary code $\mathcal{C}=\{000,111\}$. Then $\mathcal{C}^{*_1}=\{ 00,11\}$ is a $[2,1,2]$-binary linear code.
\item
\label{ejemploper2}
  Let $\mathcal{C}=\{0000,1100,0011,1111\}$, note that it is a $[4,2,2]$-linear code. Thus the code $\mathcal{C}^{*_2}=\{ 000,100,011,111\}$ is a $[3,2,1]$-binary linear code.

\item
\label{ejemploper3}
 For $\mathcal{C}=\{00000,11110,01100,10010\}$, we know that $\mathcal{C}$ has parameters $[5,2,2]$, but the code $\mathcal{C}^{*_5}=\{ 0000,1111,0110,1001\}$ is  a $[4,2,2]$-binary linear code. Observe, that in this case $d(\mathcal{C})=d(\mathcal{C}^{*_5})$.
\end{enumerate}
\end{example}

The following known statement explain the results obtained in Example \ref{ejemploperforado}, see \cite[Theorem 1.5.1]{Huffman2003}.

\begin{theorem}\label{param_perforado}
Let $\mathcal{C}$ be an $[n,k,d]$-binary linear code and $1\leq i\leq n$. Then 
\begin{enumerate}[label=\emph{(\arabic*)}]
\item \label{thm2.1a} If $d>1$, $\mathcal{C}^{*_i}$ is an $[n-1,k,d^{*}]$-code, where $d^{*}=d-1$ if $\mathcal{C}$ has a minimum weight codeword with a nonzero $i$-th coordinate and $d^{*}=d$ otherwise.

\item \label{thm2.1b}  If $d=1$, $\mathcal{C}^{*_i}$ is an $[n-1, k, 1]$-code if $\mathcal{C}$ has no codeword of weight 1 whose
nonzero entry is in coordinate $i$; otherwise, if $k>1$, $\mathcal{C}^{*_i}$ is an $[n-1, k-1, d^*]$-code
with $d^* \geq 1$.
\end{enumerate}
\end{theorem}

\begin{remark} Suppose that $\mathcal{C}$ is an $[n, k, 1]$-binary linear code and $\boldsymbol{e}_i\in\mathcal{C}$.
\begin{enumerate}[label=(\alph*)]
\item Assume that $\mathcal{C}=\{\boldsymbol{0},\boldsymbol{e}_i\}\subset \mathbb{F}_2^n$. Then  $\mathcal{C}^{*_i}=\mathcal{O}\subset\mathbb{F}_2^{n-1}$. Therefore, by \cite[Example 2.2(2)]{Delgado2024}, $\Gamma(\mathcal{C}^{*_i})\cong H(n-1,2)$ is the Hamming graph with $2^{n-1}$ vertices. Since, the dimension of $\mathcal{C}$ is 1, then $\Gamma(\mathcal{C})$ has also $2^{n-1}$ vertices. Actually, it can be proved that $\varphi(\boldsymbol{u}+\mathcal{C})=\boldsymbol{u}^{*_i}$ is a graph homomorphism and thus $\Gamma(\mathcal{C})\cong \Gamma(\mathcal{C}^{*_i})\cong H(n-1,2)$.

\item Now we consider when $k>1$. By, Theorem \ref{param_perforado}\emph{\ref{thm2.1b}}, $\mathcal{C}^{*_i}$ has parameters $[n-1, k-1, d^*]$, with $d^*\geq 1$. Then $\Gamma(\mathcal{C})$ and $\Gamma(\mathcal{C}^{*_i})$ have both $2^{n-k}$ vertices. Furthermore, if $\boldsymbol{u}+\mathcal{C}\in V_{\mathcal{C}}$, then $\boldsymbol{u}^{*_i}+\mathcal{C}^{*_i}\in V_{\mathcal{C}^{*_i}}$. So, we can define the map $\phi(\boldsymbol{u}+\mathcal{C})=\boldsymbol{u}^{*_i}+\mathcal{C}^{*_i}$. Moreover, it can be verified that each edge of $\Gamma(\mathcal{C})$  is sent to an edge of $\Gamma(\mathcal{C}^{*_i})$. Therefore, $\Gamma(\mathcal{C})\cong \Gamma(\mathcal{C}^{*_i})$.
\end{enumerate}
\end{remark}
From the rest of this section, we study the case when $\mathcal{C}$ is an $[n,k,d]$-binary linear code with either $d>1$ or $d=1$ and $\boldsymbol{e}_i\not\in\mathcal{C}$.
\begin{proposition}\label{coset_representatives}
    Let $\mathcal{C}$ be a binary linear code, $1\leq i\leq n$ and $\boldsymbol{x}+\mathcal{C}$ and $\boldsymbol{y}+\mathcal{C}$ cosets of $\mathcal{C}$. If $\boldsymbol{x}+\mathcal{C}=\boldsymbol{y}+\mathcal{C}$, then $\boldsymbol{x}^{*_i}+\mathcal{C}^{*_i}=\boldsymbol{y}^{*_i}+\mathcal{C}^{*_i}$. Moreover, if $\boldsymbol{x}$ is a coset leader of $\boldsymbol{x}+\mathcal{C}$, then there exists a coset leader $\boldsymbol{y} \in \boldsymbol{x}+\mathcal{C}$ such that $\boldsymbol{y}^{*_i}$ is a coset leader of $\boldsymbol{x}^{*_i}+\mathcal{C}^{*_i}$.   
\end{proposition}

\begin{proof}
The first assertion follows from the fact that $(\boldsymbol{x}+\boldsymbol{y})^{*_i}=\boldsymbol{x}^{*_i}+\boldsymbol{y}^{*_i}$. In general, it always holds that
        $$\wt(\boldsymbol{y}^{*_i})\leq \wt(\boldsymbol{y}),\ \text{ for all $\boldsymbol{y}\in \mathbb{F}_2^n$}.$$
Now, since $\boldsymbol{x}$ is a coset leader of $\boldsymbol{x}+\mathcal{C}$, we know that for any $\boldsymbol{c}\in \mathcal{C}$,
        $$\wt(\boldsymbol{x})\leq \wt(\boldsymbol{x}+\boldsymbol{c}).$$
        
Next, we consider the following cases: 
        \begin{enumerate}[label=Case \arabic*., align=left] 
        \item Assume that $\wt(\boldsymbol{x})<\wt(\boldsymbol{x}+\boldsymbol{c})$ for all $\boldsymbol{c}\in \mathcal{C}$. Then $\wt(\boldsymbol{x}^{*_i})\leq \wt(\boldsymbol{x})\leq \wt(\boldsymbol{x}+\boldsymbol{c})-1\leq \wt((\boldsymbol{x}+\boldsymbol{c})^{*_i})$. Therefore, $\boldsymbol{x}^{*_i}$ is a coset leader of the coset $\boldsymbol{x}^{*_i}+\mathcal{C}^{*_i}$.

        \item Suppose that $\wt(\boldsymbol{x})=\wt(\boldsymbol{x}+\boldsymbol{c})$ for some $\boldsymbol{c}\in \mathcal{C}$. This implies that $\boldsymbol{x}+\boldsymbol{c}$ is also a coset leader of $\boldsymbol{x}+\mathcal{C}$. 

        \begin{enumerate}[label=2.\alph*)]
            \item If $i \not\in \supp(\boldsymbol{x}+\boldsymbol{c})$, then $\wt(\boldsymbol{x}+\boldsymbol{c})=\wt((\boldsymbol{x}+\boldsymbol{c})^{*_i})$. Thus, $$\wt(\boldsymbol{x}^{*_i})\leq \wt(\boldsymbol{x})=\wt(\boldsymbol{x}+\boldsymbol{c})=\wt((\boldsymbol{x}+\boldsymbol{c})^{*_i}).$$
            \item Suppose that $i\in \supp(\boldsymbol{x}+\boldsymbol{c})$; hence $\wt((\boldsymbol{x}+\boldsymbol{c})^{*_i})=\wt(\boldsymbol{x}+\boldsymbol{c})-1$. 
            \begin{enumerate}[label=(\normalfont\roman*)]
                \item Assume that $i\in \supp(\boldsymbol{x})$. Then,
                $$\wt(\boldsymbol{x}^{*_i})=\wt(\boldsymbol{x})-1=\wt(\boldsymbol{x}+\boldsymbol{c})-1=\wt((\boldsymbol{x}+\boldsymbol{c})^{*_i}),$$
                \item Finally, suppose that $i\not\in \supp(\boldsymbol{x})$. Then $i\in \supp(\boldsymbol{c})$ and 
$$\wt((\boldsymbol{x}+\boldsymbol{c})^{*_i})=\wt(\boldsymbol{x}+\boldsymbol{c})-1=\wt(\boldsymbol{x}^{*_i})-1<\wt(\boldsymbol{x}^{*_i}).$$

\end{enumerate} 
        \end{enumerate}
        \end{enumerate}
Consequently, in any case it follows that $(\boldsymbol{x}+\boldsymbol{c})^{*_i}$ (for some $\boldsymbol{c}\in \mathcal{C}$) is a coset leader of $\boldsymbol{x}^{*_i}+\mathcal{C}^{*_i}$. \qedhere
\end{proof}

\begin{proposition}\label{prop:clases_perforadas}
Let $\mathcal{C}$ be an $[n,k,d]$-binary linear code and $1\leq i\leq n$ with either $d>1$ or $d=1$ and $\boldsymbol{e}_i\not\in\mathcal{C}$. Then, for each coset $\boldsymbol{x}+\mathcal{C}$, there exists a coset $\boldsymbol{z}+\mathcal{C}$ such that $\boldsymbol{x}+\mathcal{C}\neq \boldsymbol{z}+\mathcal{C}$ and $\boldsymbol{x}^{*_i}+\mathcal{C}^{*_i}=\boldsymbol{z}^{*_i}+\mathcal{C}^{*_i}$.
\end{proposition}

\begin{proof}
    Suppose that $\boldsymbol{x}$ is a coset leader of the coset $\boldsymbol{x}+\mathcal{C}$. Since $d>1$ or $d=1$ and $\boldsymbol{e}_i\not\in\mathcal{C}$, the coset $(\boldsymbol{x}+\boldsymbol{e}_i)+\mathcal{C}$ is distinct from $\boldsymbol{x}+\mathcal{C}$. Moreover, it holds that \[\boldsymbol{x}^{*_i}+\mathcal{C}^{*_i}=(\boldsymbol{x}+\boldsymbol{e}_i)^{*_i}+\mathcal{C}^{*_i}.\qedhere\]
\end{proof}

\begin{remark}\label{remark2}
    By propositions \ref{coset_representatives} and \ref{prop:clases_perforadas}, the set of cosets of $\mathcal{C}^{*_i}$ is given by $$\mathfrak{cl}(\mathcal{C}^{*_i})=\left\{ \boldsymbol{x}^{*_i}+\mathcal{C}^{*_i}:i\in \supp(\boldsymbol{x}) \text{ and } \boldsymbol{x}+\mathcal{C}\in V_{\mathcal{C}}
        \right\}.$$

Consequently, for each $\boldsymbol{x}$ we have that two cosets of $\mathcal{C}$ ($\boldsymbol{x}+\mathcal{C}$ and $\boldsymbol{x}+\boldsymbol{e}_i+\mathcal{C}$) collapse into a single coset of $\mathcal{C}^{*_i}$. Moreover, by Theorem \ref{maintheoremarticle1}\emph{\ref{thmmain1}} and Theorem \ref{param_perforado}, we get that $|V_{\mathcal{C}^{*_i}}|=\dfrac{2^{n-k}}{2}=\dfrac{|V_{\mathcal{C}}|}{2}.$ 
\end{remark}
Now, in order to provide a characterization of $\Gamma(\mathcal{C}^{*_i})$, we need to define the following graph operation.
\begin{definition}
The contraction of a pair of vertices $v_i$ and $v_j$ in a graph $\Gamma$ is the operation that produces a graph in which the two vertices $v_i$ and $v_j$ are replaced by a single vertex $v$, such that this new vertex is adjacent to the vertices to which $v_i$ and $v_j$ were originally adjacent.
\end{definition}

This concept was studied by S. Pemmaraju and S. Skiena in \cite{pemmaraju2003computational}. We will use $\Gamma^{\circledast}/[v_i,v_j]$ to denote the contraction of vertices $v_i$ and $v_j$ in the graph $\Gamma$. The graph resulting from this operation is a simple graph, i.e., without parallel edges.  

\begin{example}\label{ejem:contraccion}
Let $\Gamma$ be the graph shown in Figure \ref{fig:antes_contraccion}. Performing the contraction on vertices $7$ and $9$ yields the graph $\Gamma^{\circledast}/[7,9]$, as shown in Figure \ref{fig:despues_contraccion}.
\begin{figure}[H]
    \centering
 \begin{subfigure}[b]{0.49\textwidth}
		\centering
		\includegraphics[scale=0.45]{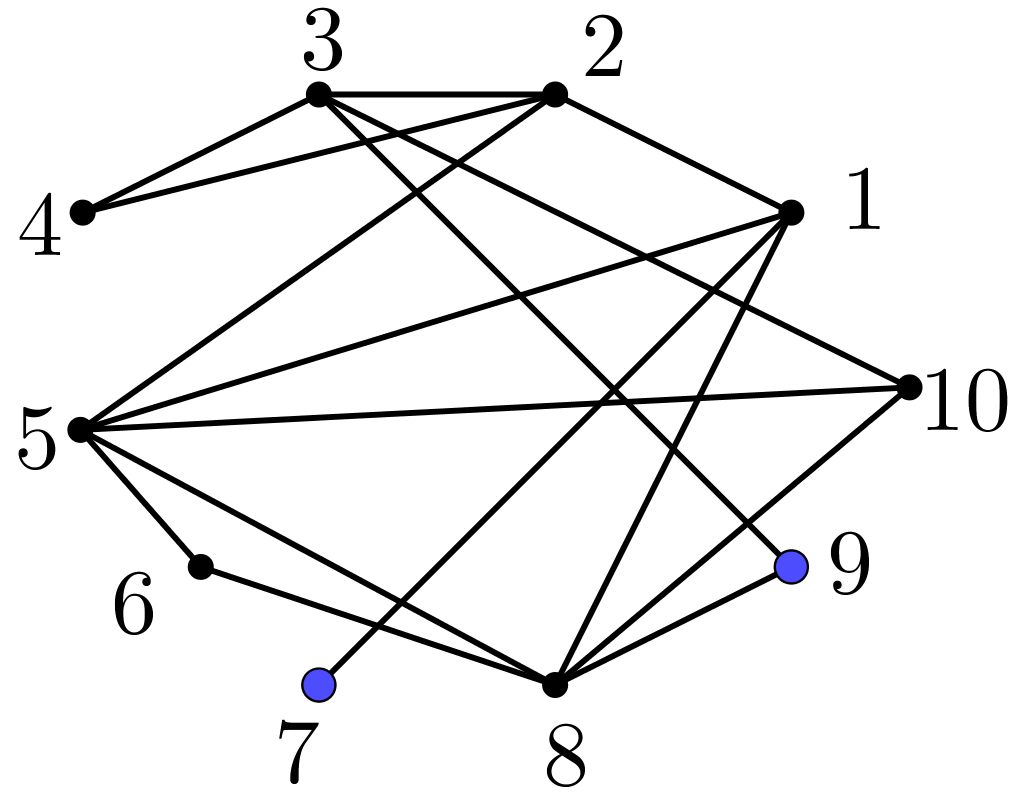}
		\caption{Graph $\Gamma$.}
		\label{fig:antes_contraccion}
	\end{subfigure}
 \begin{subfigure}[b]{0.49\textwidth}
		\centering
		\includegraphics[scale=0.45]{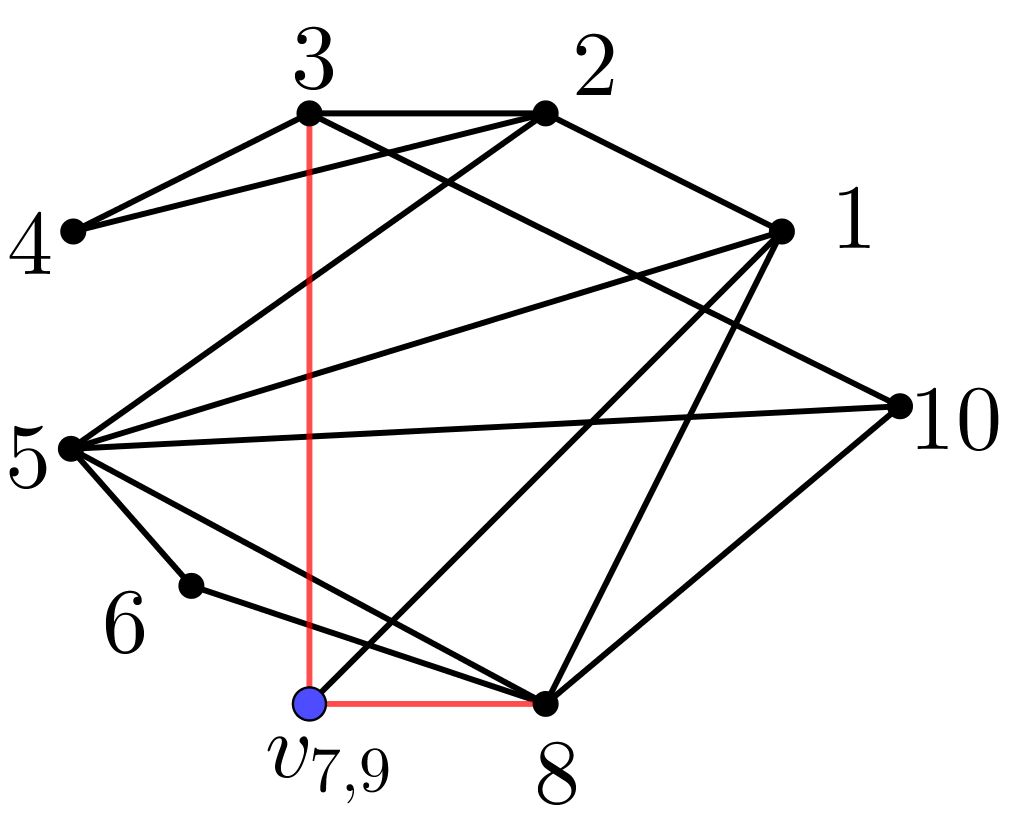}
		\caption{Graph $\Gamma^{\circledast}/[7,9]$.}
		\label{fig:despues_contraccion}
	\end{subfigure}
    \caption{Example of graph contraction.}
    \label{fig:contraccion}
\end{figure}
\end{example}
In general, given a graph $\Gamma$ and a set of vertex pairs $S=\{[v_{i_1},v_{j_1}],\ldots, [v_{i_k},v_{j_k}]\}$, we denote by $\Gamma^{\circledast}/S$ the graph resulting from successively applying vertex contraction to $\Gamma$ for each pair $[v_{i_l},v_{j_l}]$ in $S$.
\begin{example}
  For the graph $\Gamma$ in Example \ref{ejem:contraccion}, the vertex contraction $\Gamma^{\circledast}/\{[7,9],[1,3],[4,5]\}$ results in the graph presented in Figure \ref{fig:contraccion_3_parejas}.
  \begin{figure}[H]
      \centering
      \includegraphics[scale=0.45]{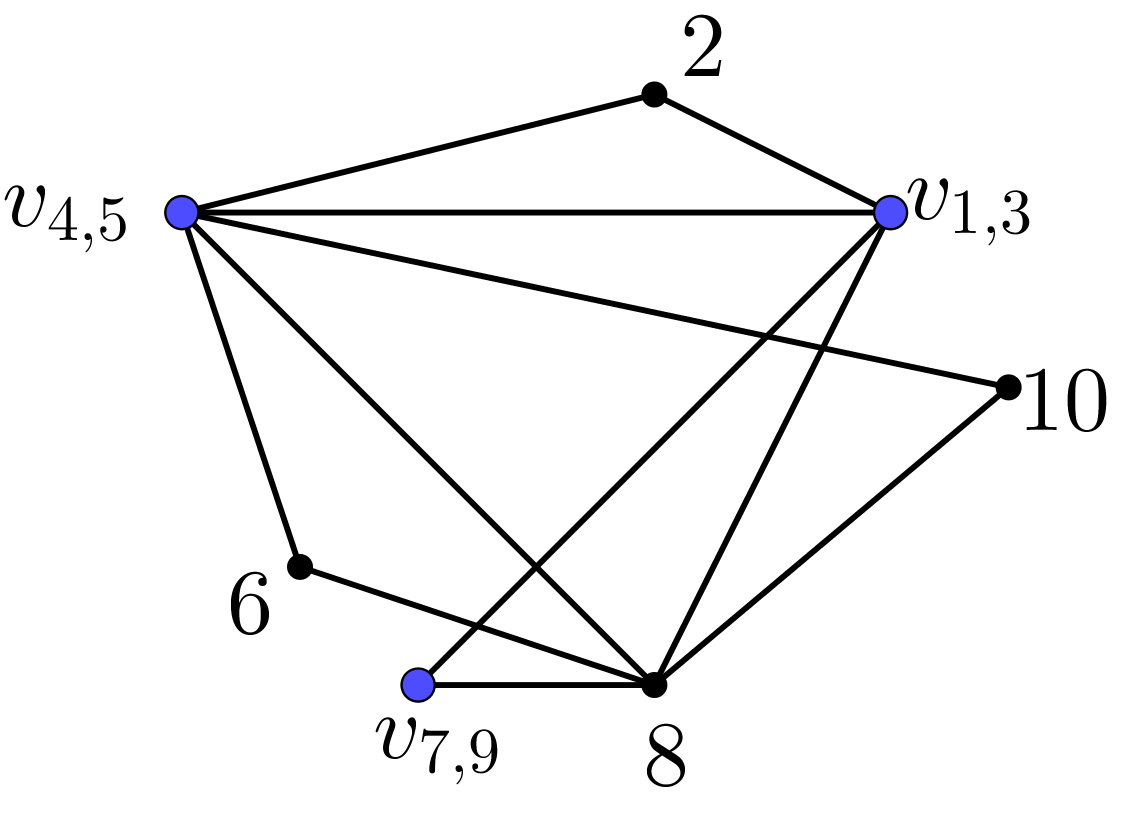}
      \caption{Graph $\Gamma^{\circledast}/\{[7,9],[1,3],[4,5]\}$.}
      \label{fig:contraccion_3_parejas}
  \end{figure}
\end{example}

\begin{theorem}
Let $\mathcal{C}$ be an $[n,k,d]$-binary linear code and $1\leq i\leq n$ with either $d>1$ or $d=1$ and $\boldsymbol{e}_i\not\in\mathcal{C}$ .  Then 
    \begin{align*}
        V_{\mathcal{C}^{*_i}}&=\Bigl\{ \boldsymbol{x}^{*_i}+\mathcal{C}^{*_i}:i\in\supp(\boldsymbol{x}) \text{ and } \boldsymbol{x}+\mathcal{C}\in V_{\mathcal{C}}\Bigr\}.\\
E_{\mathcal{C}^{*_i}}&=\Bigl\{\left\{\boldsymbol{x}^{*_i}+\mathcal{C}^{*_i},\boldsymbol{y}^{*_i}+\mathcal{C}^{*_i}\right\}: i\in \supp(\boldsymbol{x}) \text{ and } \left\{\boldsymbol{x}+\mathcal{C},\boldsymbol{y}+\mathcal{C}\right\}\in E_{\mathcal{C}}\Bigr\}.
    \end{align*}
Furthermore, $\Gamma(\mathcal{C}^{*_i})$ can be obtained from $\Gamma(\mathcal{C})$ vertex contraction respect to the set $V_i=\Bigl\{[\boldsymbol{x}+\mathcal{C},\boldsymbol{x}+\boldsymbol{e}_i+\mathcal{C}]: \boldsymbol{x}+\mathcal{C}\in V_{\mathcal{C}}\Bigr\}$; that is $$\Gamma(\mathcal{C}^{*_i})\cong\Gamma(\mathcal{C})^{\circledast}/V_i.$$
\end{theorem}

\begin{proof}
 The set of vertices $V_{\mathcal{C}^{*_i}}$ is obtained in Remark \ref{remark2}. To characterize the set of edges of $\Gamma(\mathcal{C}^{*_i})$ just rest to prove the relationship between the weights of two its vertices. Suppose that $\{\boldsymbol{x}+\mathcal{C},\boldsymbol{y}+\mathcal{C}\}\in E_{\mathcal{C}}$ with $i\in \supp(\boldsymbol{x})$. Thus $i\in \supp(\boldsymbol{x})\subset \supp(\boldsymbol{y})$ and $\wt(\boldsymbol{x})=\wt(\boldsymbol{y})-1$.  In consequence, $\supp(\boldsymbol{x}^{*_i})\subset \supp(\boldsymbol{y}^{*_i})$ and 
        $\wt(\boldsymbol{x}^{*_i})=\wt(\boldsymbol{x})-1=(\wt(\boldsymbol{y})-1)-1=\wt(\boldsymbol{y}^{*_i})-1.$ 
Therefore,  $\left\{\boldsymbol{x}^{*_i}+\mathcal{C}^{*_i},\boldsymbol{y}^{*_i}+\mathcal{C}^{*_i}\right\}\in E_{\mathcal{C}^{*_i}}$. Also, if $\left\{\boldsymbol{x}^{*_i}+\mathcal{C}^{*_i},\boldsymbol{y}^{*_i}+\mathcal{C}^{*_i}\right\}$ is an edge of $\Gamma(\mathcal{C}^{*_i})$, then by Remark \ref{remark2},  $\boldsymbol{x}+\mathcal{C}$ and $\boldsymbol{y}+\mathcal{C}$ are vertices of $\Gamma(\mathcal{C})$ such that we have that 
$i\in (\supp(\boldsymbol{x})\cap \supp(\boldsymbol{y}))$. It implies that $\left\{\boldsymbol{x}+\mathcal{C},\boldsymbol{y}+\mathcal{C}\right\}\in E_{\mathcal{C}}$.
				
The last assertion follows from Proposition \ref{prop:clases_perforadas}, since the graph $\Gamma(\mathcal{C}^{*_i})$ is obtained by successively contracting the pairs of vertices $[\boldsymbol{x}+\mathcal{C},\boldsymbol{x}+\boldsymbol{e}_i+\mathcal{C}]$ in the graph $\Gamma(\mathcal{C})$, where $\boldsymbol{x}+\mathcal{C}\in V_{\mathcal{C}}$.
\end{proof}

\begin{example}
   For the binary linear code of Example \ref{ejemploperforado}\ref{ejemploper3}, we know that $\mathcal{C}^{*_5}=\{ 0000,1111,0110,1001\}$ is a $[4,2,2]$-binary linear code. In Figure \ref{fig:ejemplo_perforado}, we show the process to obtain the Hase diagram of $\mathcal{C}^{*_5}$ from $\Gamma(\mathcal{C})$ using the vertex contraction.
\end{example}

     \begin{figure}[H]
    \begin{subfigure}[b]{0.49\textwidth}
    	\centering
    	\includegraphics[scale=0.5]{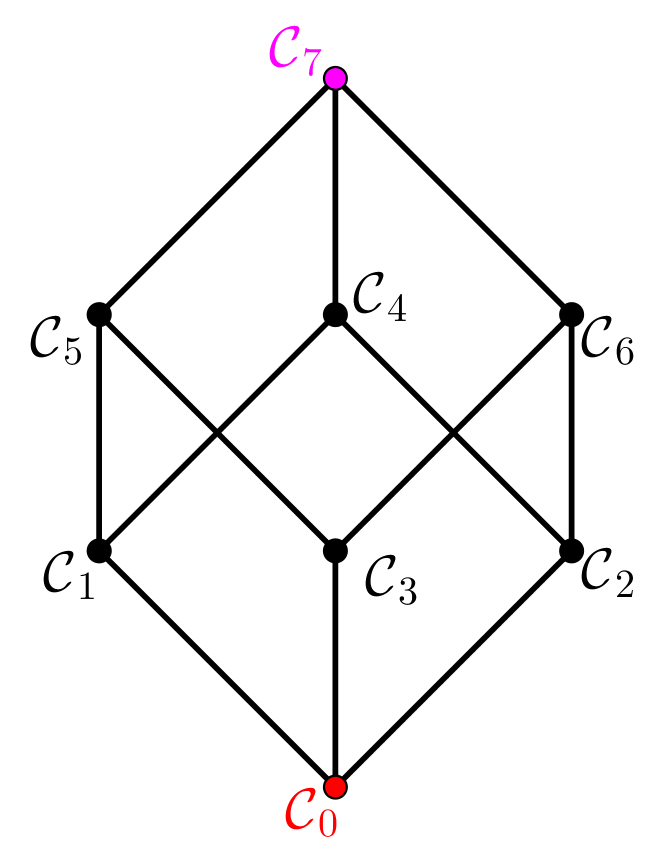} 
    	\caption{Hasse diagram of $\mathcal{C}$.}
    	\label{fig:grafoor15}
    	\end{subfigure}
 \begin{subfigure}[b]{0.49\textwidth}
    		\centering
    		\includegraphics[scale=0.5]{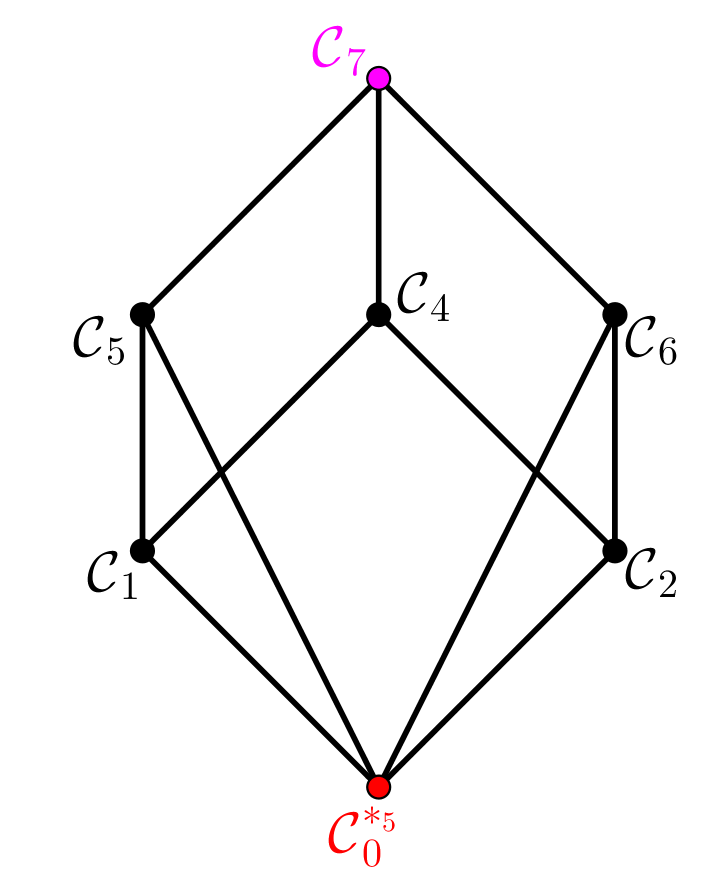}
    		\caption{Graph $\Gamma(\mathcal{C})^{\circledast}/\{[\mathcal{C}_0,\mathcal{C}_3]\}$.}
    	\end{subfigure}
		\end{figure}
	
		\begin{figure}[H]
    \begin{subfigure}[b]{0.49\textwidth}
    	\centering
    	\includegraphics[scale=0.5]{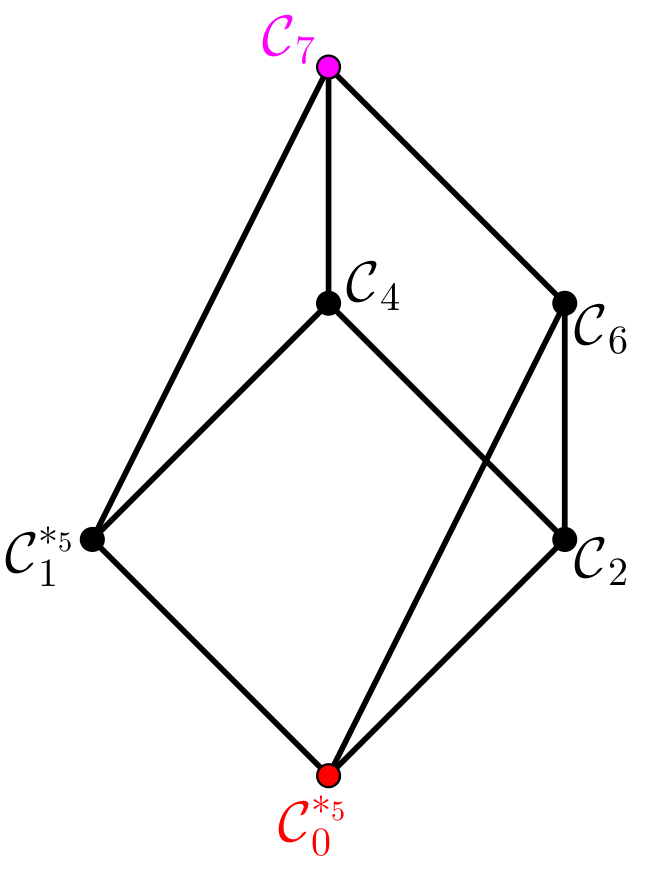} 
    	\caption{$\Gamma(\mathcal{C})^{\circledast}/\{[\mathcal{C}_0,\mathcal{C}_3],[\mathcal{C}_1,\mathcal{C}_5]\}$.}
    	\end{subfigure}
   \begin{subfigure}[b]{0.49\textwidth}
    	\centering
    	\includegraphics[scale=0.5]{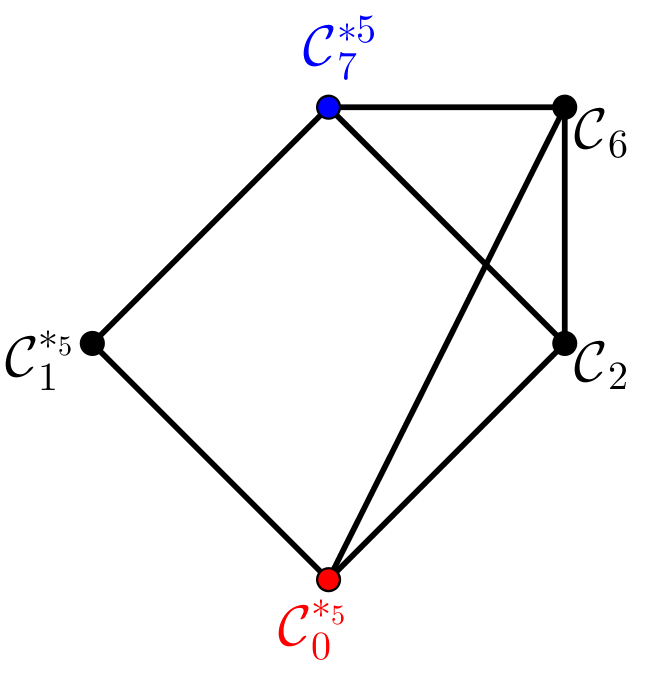} 
    	\caption{$\Gamma(\mathcal{C})^{\circledast}/\{[\mathcal{C}_0,\mathcal{C}_3],[\mathcal{C}_1,\mathcal{C}_5],[\mathcal{C}_7,\mathcal{C}_4]\}$.}
    	\end{subfigure}
    \begin{subfigure}[b]{1\textwidth}
    	\centering
    	\includegraphics[scale=0.5]{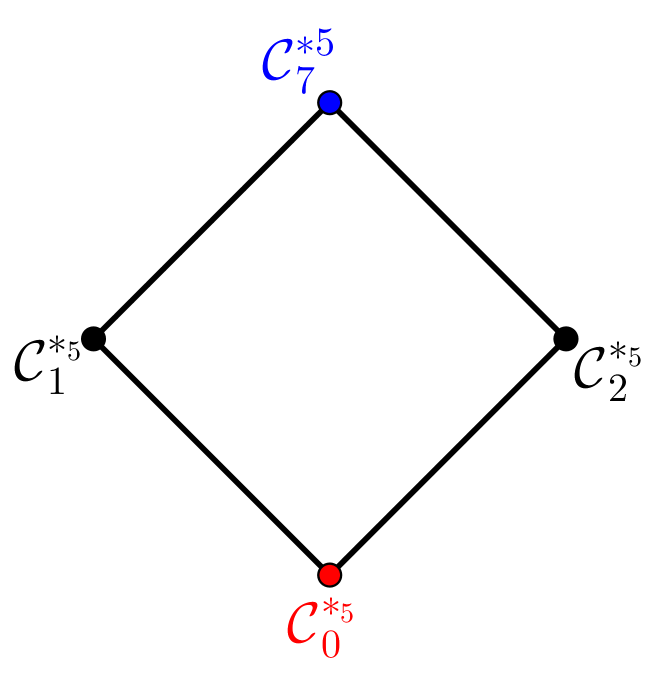} 
    	\caption{Graph $\Gamma(\mathcal{C}^{*_{5}})=\Gamma(\mathcal{C})^{\circledast}/\{[\mathcal{C}_0,\mathcal{C}_3],[\mathcal{C}_1,\mathcal{C}_5],[\mathcal{C}_7,\mathcal{C}_4],[\mathcal{C}_2,\mathcal{C}_6]\}$.}
    	\label{fig:grafoor15_perf}
    	\end{subfigure}
			\caption{Hasse diagram of $\mathcal{C}^{*_{5}}$ obtained from $\Gamma(\mathcal{C})$.}\label{fig:ejemplo_perforado}
    \end{figure}

\subsection{Extending a code}
 \begin{definition}[Extending a code]
Let $\mathcal{C}$ be an $[n,k,d]$-binary linear code. The process of adding one or more coordinate positions to the codewords in $\mathcal{C}$ is referred to as \emph{extending a code}. For instance, the \emph{extended code} of $\mathcal{C}$ is  
\begin{equation}
\widehat{\mathcal{C}}=\left\{c_1\ldots c_nc_{n+1}: c_1\ldots c_n\in \mathcal{C}\text{ with } \sum_{k=1}^{n+1}c_k=0  \right\}.
\label{eq:extending}
\end{equation}
\end{definition}
In this case, the extra coordinate $c_{n+1}$ added to the codeword $c_1\ldots c_n$ is called the \emph{parity-check coordinate}. The new parameters of $\widehat{\mathcal{C}}$ defined by \eqref{eq:extending} are given in the next result, see \cite[Section 1.5.2]{Huffman2003}.
 \begin{theorem}\label{teoextendido}
Let $\mathcal{C}$ be an $[n,k,d]$-binear linear code. Then $\widehat{\mathcal{C}}$ is an $[n+1,k,\widehat{d}]$-binary linear code, where either $\widehat{d}=d$ or $\widehat{d}=d+1$.
\end{theorem}
From theorems \ref{maintheoremarticle1} and \ref{teoextendido}, we obtain the following result.
\begin{corollary}
\label{corcardinalextendido}
    Let $\widehat{\mathcal{C}}$ be the extended code of an $[n,k,d]$-binary linear code $\mathcal{C}$. Then, $|V_{\widehat{\mathcal{C}}}|=2^{(n+1)-k}$. 
\end{corollary}

\begin{theorem}[{\cite[Theorem 5.1.9]{Ling2004}}]\label{teor:sindrome_extendido}
    Let $\mathcal{C}$ be a binary linear code and  $H$ be a parity-check matrix of $\mathcal{C}$. Then, 
$$\widehat{H} =
\left(
  \begin{array}{ccc|c}
  &  &  & 0 \\
  & H & & \vdots \\
  &  &  & 0 \\
  \hline
  1& \cdots &1 & 1
\end{array} \right)$$
is a parity-check matrix for $\widehat{\mathcal{C}}$.
\end{theorem}

\begin{example}\quad
\label{extendido}

    \begin{enumerate}[label=(\alph*)]
        \item 
        \label{extendidoa} In Figure \ref{fig:ext1}, we show the Hass diagram for the extended codes of the codes in examples \ref{ejemploextendido}\ref{ejemplo2ext} and \ref{ejemploextendido}\ref{ejemplo1ext}. Note that in this case these graphs are the same (isomorphics). 

        \begin{figure}[H]
		\centering
		\includegraphics[scale=0.75]{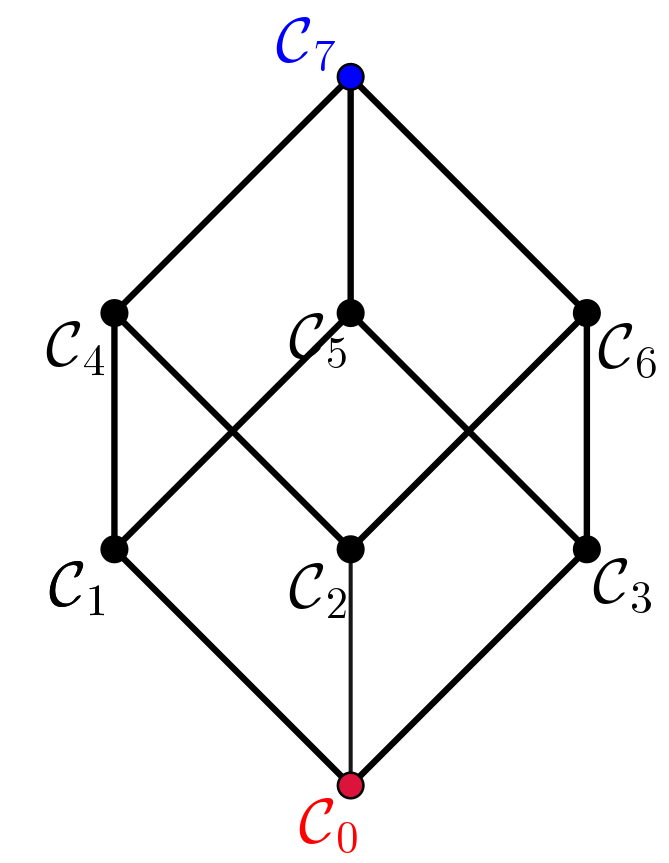}
		\caption{Graph of the extended code of examples \ref{ejemploextendido}\ref{ejemplo2ext} and \ref{ejemploextendido}\ref{ejemplo1ext}.}
		\label{fig:ext1}
	\end{figure}

       \item \label{extendidob} The graph of the extended code from Example \ref{ejemploextendido}\ref{ejemplo3ext} is shown in Figure \ref{fig:ext2}.

       \begin{figure}[H]
		\centering
		\includegraphics[scale=0.7]{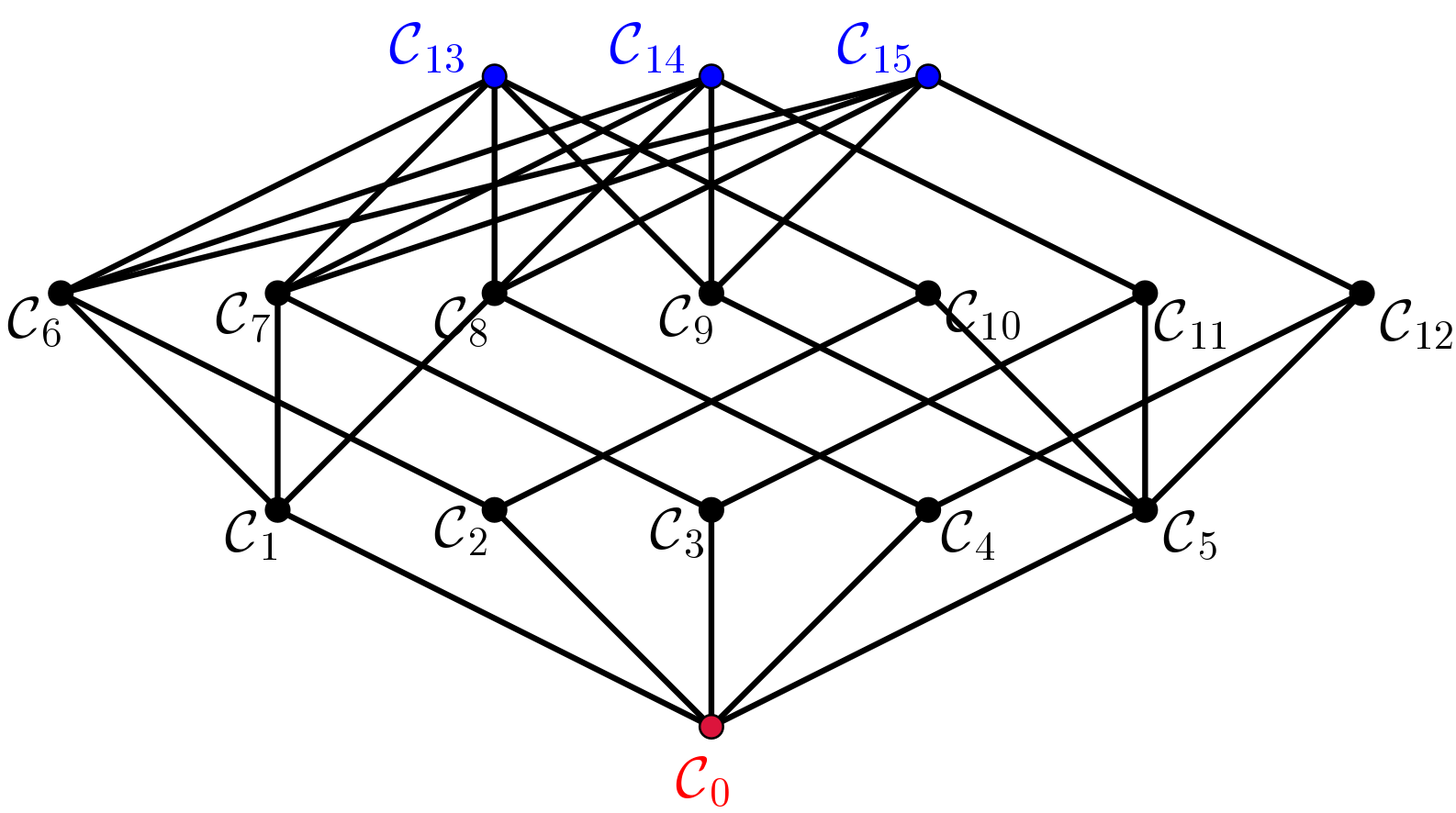}
		\caption{Graph of the extended code of Example \ref{ejemploextendido}\ref{ejemplo3ext}.}
		\label{fig:ext2}
	\end{figure}
        
    \end{enumerate}
\end{example}
Now, by means of some results, we present a method to construct the graph of the extended code without computing the extended code explicitly. For $\boldsymbol{x}\in \mathbb{F}_2^n$, where $\boldsymbol{x}=x_1x_2\ldots x_n$, consider the vector $\overline{\boldsymbol{x}}=x_1x_2\ldots x_nx_{n+1}$, where either $x_{n+1}=0$ or $x_{n+1}=1$. In particular, we denote $\ddot{\boldsymbol{x}}=x_1\ldots x_n0 $ and $ \widetilde{{\boldsymbol{x}}}=x_1\ldots x_n1$.   

    \begin{proposition}
 If $\boldsymbol{x}\not\in \boldsymbol{y}+\mathcal{C}$, then $\overline{\boldsymbol{x}}\notin \overline{\boldsymbol{y}}+\widehat{\mathcal{C}}$. Moreover, if $\boldsymbol{x}$ is a coset leader of $\boldsymbol{x}+\mathcal{C}$, then $\overline{\boldsymbol{x}}$ is a coset leader of $\overline{\boldsymbol{x}}+\widehat{\mathcal{C}}$.
\end{proposition}

\begin{proof}
    First, by Theorem \ref{teor:sindrome_extendido}, we compute the syndrome of $\overline{\boldsymbol{x}}\in \mathbb{F}_2^{n+1}$.
    \begin{align*}
\Syn(\overline{\boldsymbol{x}})&=\overline{\boldsymbol{x}}\widehat{H}^T=\overline{\boldsymbol{x}} \left(
  \begin{array}{ccc|c}
   &  &  & 1 \\
   & H^T & & \vdots \\
   &  &  & 1 \\
  \hline
  0& \ldots &0 & 1
\end{array} \right)= \left( \boldsymbol{x} H^T, \sum_{i=1}^{n+1} x_i \right)\\
&=  \left( \Syn(\boldsymbol{x}), \sum_{i=1}^{n+1} x_i \right).
      \end{align*} 
Similarly, $\Syn(\overline{\boldsymbol{y}})=\left( \Syn(\boldsymbol{y}), \sum_{i=1}^{n+1} y_i \right)$. Now, by hypothesis, $\boldsymbol{x}\notin \boldsymbol{y}+\mathcal{C}$; that is by \cite[Theorem 1.11.5]{Huffman2003}, $\Syn(\boldsymbol{x})\neq \Syn(\boldsymbol{y})$. Therefore, $\Syn(\overline{\boldsymbol{x}})\neq \Syn(\overline{\boldsymbol{y}})$, which implies that $\overline{\boldsymbol{x}}\notin \overline{\boldsymbol{y}}+\widehat{\mathcal{C}}$.

Next, we must show that $\overline{\boldsymbol{x}}$ is a coset leader of its coset; provided $\boldsymbol{x}$ is a coset leader of $\boldsymbol{x}+\mathcal{C}$. To achieve this, we need to prove that $\wt(\overline{\boldsymbol{x}})\leq \wt(\overline{\boldsymbol{x}}+\widehat{\boldsymbol{c}})$, for all $\widehat{\boldsymbol{c}}\in \widehat{\mathcal{C}}$. Suppose that $\boldsymbol{x}$ is the unique coset leader of its coset; it means that
\begin{equation}
\label{eq3.2}
    \wt(\boldsymbol{x})< \wt(\boldsymbol{x}+\boldsymbol{c}),\ \forall \boldsymbol{c}\in \mathcal{C}.
\end{equation}
Now, we consider the following cases:
\begin{enumerate}
    \item Suppose that $\overline{\boldsymbol{x}}=\ddot{\boldsymbol{x}}$. Hence, $\wt(\overline{\boldsymbol{x}})=\wt(\boldsymbol{x})$. Thus, by \eqref{eq3.2}, $\wt(\overline{\boldsymbol{x}})=\wt(\boldsymbol{x})<\wt(\boldsymbol{x}+\boldsymbol{c})$. Now, consider two subcases: 
\begin{enumerate}[label=(\normalfont\roman*)]
    \item If $\widehat{\boldsymbol{c}}=\ddot{\boldsymbol{c}}$, then $\overline{\boldsymbol{x}}+\widehat{\boldsymbol{c}}=\ddot{\boldsymbol{x}}+\ddot{\boldsymbol{c}}=(\boldsymbol{x}+\boldsymbol{c})0$, which implies that $\wt(\overline{\boldsymbol{x}}+\widehat{\boldsymbol{c}})=\wt(\boldsymbol{x}+\boldsymbol{c})$. It follows that $\wt(\overline{\boldsymbol{x}})<\wt(\overline{\boldsymbol{x}}+\widehat{\boldsymbol{c}})$.

    \item If $\widehat{\boldsymbol{c}}=\widetilde{\boldsymbol{c}}$, then $\overline{\boldsymbol{x}}+\widehat{\boldsymbol{c}}=\ddot{\boldsymbol{x}}+\widetilde{\boldsymbol{c}}=(\boldsymbol{x}+\boldsymbol{c})1$. Then $\wt(\overline{\boldsymbol{x}}+\widehat{\boldsymbol{c}})=\wt(\boldsymbol{x}+\boldsymbol{c})+1$. Accordingly, from \eqref{eq3.2}, $ \wt(\overline{\boldsymbol{x}})<\wt(\boldsymbol{x}+\boldsymbol{c})+1= \wt(\overline{\boldsymbol{x}}+\widehat{\boldsymbol{c}})$.
    
\end{enumerate}

\item
\label{caso2}
Now, assume that $\overline{\boldsymbol{x}}=\widetilde{\boldsymbol{x}}$. Then,  $\wt(\overline{\boldsymbol{x}})=\wt(\boldsymbol{x})+1$.
Again, we study two subcases:
\begin{enumerate}[label=(\normalfont\roman*)]
    \item If $\widehat{\boldsymbol{c}}=\ddot{\boldsymbol{c}}$, then $\overline{\boldsymbol{x}}+\widehat{\boldsymbol{c}}=\widetilde{\boldsymbol{x}}+\ddot{\boldsymbol{c}}=(\boldsymbol{x}+\boldsymbol{c})1$, which implies that $\wt(\overline{\boldsymbol{x}}+\widehat{\boldsymbol{c}})=\wt(\boldsymbol{x}+\boldsymbol{c})+1$. Thus, by \eqref{eq3.2}, $ \wt(\overline{\boldsymbol{x}})=\wt(\boldsymbol{x})+1<\wt(\boldsymbol{x}+\boldsymbol{c})+1= \wt(\overline{\boldsymbol{x}}+\widehat{\boldsymbol{c}})$.

    \item
    \label{casonl}
    If $\widehat{\boldsymbol{c}}=\widetilde{\boldsymbol{c}}$, then $\overline{\boldsymbol{x}}+\widehat{\boldsymbol{c}}=\widetilde{\boldsymbol{x}}+\widetilde{\boldsymbol{c}}=(\boldsymbol{x}+\boldsymbol{c})0$, which implies that $\wt(\overline{\boldsymbol{x}}+\widehat{\boldsymbol{c}})=\wt(\boldsymbol{x}+\boldsymbol{c})$. Thus, \begin{equation}\label{eq:posible_nuevo_lider}
    \wt(\overline{\boldsymbol{x}})=\wt(\boldsymbol{x})+1\leq \wt(\boldsymbol{x}+\boldsymbol{c})=\wt(\overline{\boldsymbol{x}}+\widehat{\boldsymbol{c}}).\qedhere
     \end{equation}
\end{enumerate}
\end{enumerate}
\end{proof}

\begin{remark}
    The only case in which new coset leaders can appear in the cosets of the extended code is in case \ref{caso2}\ref{casonl}. This is because a vector that is not a coset leader in the original code can become a coset leader in the extended code when equality holds in \eqref{eq:posible_nuevo_lider}; that is when $\wt(\boldsymbol{x}+\boldsymbol{c})=\wt(\boldsymbol{x})+1$. We recall that $\boldsymbol{x}\cap \boldsymbol{y} =(x_1y_1,x_2y_2,\ldots, x_ny_n)$, for all $\boldsymbol{x},\boldsymbol{y}\in \mathbb{F}_2^n$. Thus by \cite[Theorem 1.4.3]{Huffman2003}, we know that $$\wt(\boldsymbol{x}+\boldsymbol{c})=\wt(\boldsymbol{x})+\wt(\boldsymbol{c})-2\wt(\boldsymbol{x}\cap \boldsymbol{c}),$$
the extended codes that exhibit new coset leaders are precisely those containing codewords $\boldsymbol{c}$ such that $\wt(\boldsymbol{c})=2\wt(\boldsymbol{x}\cap \boldsymbol{c})+1$, for some coset leader $\boldsymbol{x}$ of some coset of $\mathcal{C}$.
\end{remark}

\begin{corollary}
\label{cor3.5}
    If $\boldsymbol{x}\in\mathbb{F}_2^n$ is a coset leader of $\boldsymbol{x}+\mathcal{C}$, then $\ddot{\boldsymbol{x}}$ and $\widetilde{\boldsymbol{x}}$ are coset leaders of distinct cosets of $\widehat{\mathcal{C}}$. Moreover, all cosets of $\widehat{\mathcal{C}}$ are of the form $\overline{\boldsymbol{x}}+\widehat{\mathcal{C}}$, where $\boldsymbol{x}$ is a coset leader of a coset of $\mathcal{C}$.
\end{corollary}

\begin{proof}
    Let $\boldsymbol{x}\in \mathbb{F}_2^n$ be an arbitrary vector. Then $\ddot{\boldsymbol{x}},\widetilde{\boldsymbol{x}}\in \mathbb{F}_2^{n+1}$ and  
\begin{align}\label{eq:sindrome_extended}
\begin{split}
		\Syn(\ddot{\boldsymbol{x}})=(\Syn(\boldsymbol{x}),0),\\
    \Syn(\widetilde{\boldsymbol{x}})=(\Syn(\boldsymbol{x}),1).
\end{split}
\end{align}
    Thus, $\Syn(\ddot{\boldsymbol{x}})\neq \Syn(\widetilde{\boldsymbol{x}})$. Therefore, $\ddot{\boldsymbol{x}}\notin \widetilde{\boldsymbol{x}}+\mathcal{C}$. This implies that for each coset leader $\boldsymbol{x}$ of a coset of $\mathcal{C}$, we obtain two distinct cosets of $\widehat{\mathcal{C}}$, namely $\ddot{\boldsymbol{x}}+\widehat{\mathcal{C}}$ and $\widetilde{\boldsymbol{x}}+\widehat{\mathcal{C}}$, where $\ddot{\boldsymbol{x}}$ and $\widetilde{\boldsymbol{x}}$ are leaders of their respective cosets. Consequently, $|\mathfrak{cl}(\widehat{\mathcal{C})}|\geq 2|\mathfrak{cl}(\mathcal{C})|=2\cdot (2^{n-k})=2^{n+1-k}$. Furthermore, by Corollary \ref{corcardinalextendido}, it holds that $|\mathfrak{cl}(\widehat{\mathcal{C}})|=2^{n+1-k}$.
\end{proof}
In this way, in the following result, we provide a partial characterization of the vertex set and the edge set of the Hasse diagram of $\widehat{\mathcal{C}}$.
\begin{theorem}\label{teoExtendido}
   Let $\mathcal{C}$ be a linear code such that there is not a codeword $\boldsymbol{c}\in\mathcal{C}$  
	\begin{equation}
	\wt(\boldsymbol{c})=2\wt(\boldsymbol{x}\cap \boldsymbol{c})+1, \text{ for some }\boldsymbol{x}+\mathcal{C}\in \mathfrak{cl}(\mathcal{C}).
	\label{eq:star}
	\end{equation}
Then
    \begin{align*}
V_{\widehat{\mathcal{C}}}&=\left\{\ddot{\boldsymbol{x}}+\widehat{\mathcal{C}},\widetilde{\boldsymbol{x}}+\widehat{\mathcal{C}}:\boldsymbol{x}+\mathcal{C}\in V_{\mathcal{C}}\right\},\\
    E_{\widehat{\mathcal{C}}}&=\left\{\left\{\ddot{\boldsymbol{x}}+\widehat{\mathcal{C}},\ddot{\boldsymbol{y}}+\widehat{\mathcal{C}}\right\}, \left\{\widetilde{\boldsymbol{x}}+\widehat{\mathcal{C}},\widetilde{\boldsymbol{y}}+\widehat{\mathcal{C}}\right\}:\left\{\boldsymbol{x}+\mathcal{C},\boldsymbol{y}+\mathcal{C}\right\}\in E_{\mathcal{C}}\right\}\\ 
&\qquad\bigcup \left\{\left\{\ddot{\boldsymbol{x}}+\widehat{\mathcal{C}},\widetilde{\boldsymbol{x}}+\widehat{\mathcal{C}}\right\}: \boldsymbol{x}+\mathcal{C}\in V_{\mathcal{C}}\right\}.
    \end{align*}
\end{theorem}

\begin{proof}
    By Corollary \ref{cor3.5}, the cosets of $\widehat{\mathcal{C}}$ have the form $\ddot{\boldsymbol{x}}+\mathcal{C}$ and $\widetilde{\boldsymbol{x}}+\mathcal{C}$; we get $$V_{\widehat{\mathcal{C}}}=\left\{\ddot{\boldsymbol{x}}+\widehat{\mathcal{C}},\widetilde{\boldsymbol{x}}+\widehat{\mathcal{C}}:\boldsymbol{x}+\mathcal{C}\in V_{\mathcal{C}}\right\}.$$
     
Now, to determine the edge set of $\Gamma_{\widehat{\mathcal{C}}}$, suppose that $\boldsymbol{x}$ and $\boldsymbol{y}$ are coset leaders of cosets of $\mathcal{C}$. It is easy to verify that, if $\left\{\boldsymbol{x}+\mathcal{C},\boldsymbol{y}+\mathcal{C}\right\}\in E_{\mathcal{C}}$, then $\left\{\ddot{\boldsymbol{x}}+\widehat{\mathcal{C}},\ddot{\boldsymbol{y}}+\widehat{\mathcal{C}}\right\}, \left\{\widetilde{\boldsymbol{x}}+\widehat{\mathcal{C}},\widetilde{\boldsymbol{y}}+\widehat{\mathcal{C}}\right\}$ are edges of $\Gamma_{\widehat{\mathcal{C}}}$. Also, if two vertices $\boldsymbol{x}+\mathcal{C}$ and $\boldsymbol{y}+\mathcal{C}$ in $\Gamma_{\mathcal{C}}$ are not neighbors, then $\supp(\boldsymbol{x})\not\subset\supp(\boldsymbol{y})$  or $\wt(\boldsymbol{x})\neq\wt(\boldsymbol{y})-1$. Thus it can be proved that, the cosets of $\ddot{\boldsymbol{x}}, \ddot{\boldsymbol{y}}, \widetilde{\boldsymbol{x}}$ and  $\widetilde{\boldsymbol{y}}$ cannot be connected in the new graph. Moreover, for any $\boldsymbol{x}+\mathcal{C}\in V_{\mathcal{C}}$ we get that $\supp(\ddot{\boldsymbol{x}})\subset \supp(\widetilde{\boldsymbol{x}})$ and $\wt(\ddot{\boldsymbol{x}})=\wt(\widetilde{\boldsymbol{x}})-1$. Therefore,  $\left\{\ddot{\boldsymbol{x}}+\widehat{\mathcal{C}},\widetilde{\boldsymbol{x}}+\widehat{\mathcal{C}}\right\}$ is an edge of $\Gamma_{\widehat{\mathcal{C}}}$.
\end{proof}
Subsequently, we define a graph based on $\Gamma(\mathcal{C})$ that will be useful to characterize $\Gamma(\widehat{\mathcal{C}})$.
\begin{definition}
    Let $\Gamma(\mathcal{C})$ be the graph associated with the binary linear code $\mathcal{C}$. Consider the graph $\widehat{\Gamma}(\mathcal{C})$ whose vertex and edge sets are given by
    \begin{align*}
        V(\widehat{\Gamma}(\mathcal{C}))&=\left\{\widehat{\boldsymbol{x}}+\widehat{\mathcal{C}}:\boldsymbol{x}+\mathcal{C}\in V_{\mathcal{C}}\right\},\\
        E(\widehat{\Gamma}(\mathcal{C}))&=\left\{\left\{\widehat{\boldsymbol{x}}+\widehat{\mathcal{C}}, \widehat{\boldsymbol{y}}+\widehat{\mathcal{C}}\right\}:\left\{\boldsymbol{x}+\mathcal{C}, \boldsymbol{y}+\mathcal{C}\right\}\in E_{\mathcal{C}}\right\}.
    \end{align*} 
\end{definition}

We call the graph, $\widehat{\Gamma}(\mathcal{C})$ \emph{as the extended graph of $\Gamma(\mathcal{C})$}. From the preceding definition, we obtain the following result.

\begin{proposition}
\label{cor3.9}
  Let $\mathcal{C}$ be a binary linear code. Then $\Gamma(\mathcal{C})\cong \widehat{\Gamma}(\mathcal{C})$.
\end{proposition}
\begin{proof}
    Consider the function $\varphi: V_{\mathcal{C}}\rightarrow V(\widehat{\Gamma}(\mathcal{C}))$ given by $\varphi(\boldsymbol{x}+\mathcal{C})=\widehat{\boldsymbol{x}}+\widehat{\mathcal{C}}$. Next, we show that $\varphi$ is a graph isomorphism:

    \begin{enumerate}
        \item $\varphi$ is well-defined. If $\boldsymbol{x}+\mathcal{C}=\boldsymbol{y}+\mathcal{C}$, then we must prove that $\varphi(\boldsymbol{x}+\mathcal{C})=\varphi(\boldsymbol{y}+\mathcal{C})$; that is, $\widehat{\boldsymbol{x}}+\widehat{\mathcal{C}}=\widehat{\boldsymbol{y}}+\widehat{\mathcal{C}}$ if and only if $\widehat{\boldsymbol{x}}-\widehat{\boldsymbol{y}}\in\widehat{\mathcal{C}}$. Since $\boldsymbol{x}-\boldsymbol{y}\in \mathcal{C}$ and $\widehat{\mathcal{C}}$ is a binary linear code, it follows that $\widehat{\boldsymbol{x}}-\widehat{\boldsymbol{y}}=\widehat{\boldsymbol{x}-\boldsymbol{y}}\in \widehat{\mathcal{C}}$. 

    \item $\varphi$ is bijective. Indeed, since $|V_{\mathcal{C}}|=|V(\widehat{\Gamma}(\mathcal{C}))|$, it suffices to show that $\varphi$ is surjective (or injective). From the definition, it is clear that $\varphi$ is surjective. 

    \item Finally, from the definition given for the edges of $\widehat{\Gamma}(\mathcal{C})$, it follows that $\varphi$ preserves adjacency of vertices since 
    \[\left\{\boldsymbol{x}+\mathcal{C},\boldsymbol{y}+\mathcal{C}\right\}\in E_{\mathcal{C}} \text{ if and only if } \left\{ \widehat{\boldsymbol{x}}+\widehat{\mathcal{C}},\widehat{\boldsymbol{y}}+\widehat{\mathcal{C}}\right\}\in E(\widehat{\Gamma}(\mathcal{C})).\qedhere\] 
    \end{enumerate}
\end{proof}

We now show that there are two subgraphs of $\Gamma(\widehat{\mathcal{C}})$ that are isomorphic to $\Gamma(\mathcal{C})$. To this end, consider the graphs $\Gamma_0(\widehat{\mathcal{C}})=(V_0,E_0)$ and $\Gamma_1(\widehat{\mathcal{C}})=(V_1,E_1)$, which are defined as follows:
\begin{align*}
    V_0&=\left\{\ddot{\boldsymbol{x}}+\widehat{\mathcal{C}}:\boldsymbol{x}+\mathcal{C}\in V_{\mathcal{C}}\right\},\quad E_0=\left\{\left\{\ddot{\boldsymbol{x}}+\widehat{\mathcal{C}},\ddot{\boldsymbol{y}}+\widehat{\mathcal{C}}\right\}:\left\{\boldsymbol{x}+\mathcal{C},\boldsymbol{y}+\mathcal{C}\right\}\in E_{\mathcal{C}}\right\},\\
    V_1&=\left\{\widetilde{\boldsymbol{x}}+\widehat{\mathcal{C}}:\boldsymbol{x}+\mathcal{C}\in V_{\mathcal{C}}\right\}, \quad E_1=\left\{\left\{\widetilde{\boldsymbol{x}}+\widehat{\mathcal{C}},\widetilde{\boldsymbol{y}}+\widehat{\mathcal{C}}\right\}:\left\{\boldsymbol{x}+\mathcal{C},\boldsymbol{y}+\mathcal{C}\right\}\in E_{\mathcal{C}}\right\}.
\end{align*}

\begin{corollary}
Let $\mathcal{C}$ be a binary linear code. Then, 
$$\Gamma(\mathcal{C})\cong \Gamma_0(\widehat{\mathcal{C}})\cong \Gamma_1(\widehat{\mathcal{C}}).$$
That is, $\Gamma(\widehat{\mathcal{C}})$ contains two subgraphs isomorphic to $\Gamma(\mathcal{C})$.
\end{corollary}

\begin{proof}
It is enough to consider the function $$\varphi_0:V_{\mathcal{C}}\rightarrow V_0\subseteq V_{\widehat{\mathcal{C}}}$$ given by $\varphi_0(\boldsymbol{x}+\mathcal{C})=\ddot{\boldsymbol{x}}+\widehat{\mathcal{C}}$.

The function $\varphi_0$ is well-defined since if $\boldsymbol{x}+\mathcal{C}=\boldsymbol{y}+\mathcal{C}$, we have that $\boldsymbol{x}-\boldsymbol{y}\in\mathcal{C}$. Thus \eqref{eq:sindrome_extended}  and \cite[Theorem 1.11.5]{Huffman2003} imply that $(\boldsymbol{x}-\boldsymbol{y})0=\ddot{\boldsymbol{x}}-\ddot{\boldsymbol{y}}\in \widehat{\mathcal{C}}$. Therefore, $\varphi_0(\boldsymbol{x}+\mathcal{C})=\varphi_0(\boldsymbol{y}+\mathcal{C})$. 

   Clearly, $\varphi_0$ is surjective and $|V|=|V_0|$; hence, $\varphi_0$ is bijective. Moreover, by the definition of $\varphi_0$ and $\Gamma_0(\widehat{\mathcal{C}})$, it holds that 
   $$\left\{\varphi_0(\boldsymbol{x}+\mathcal{C}),\varphi_0(\boldsymbol{y}+\mathcal{C})\right\}=\left\{\ddot{\boldsymbol{x}}+\widehat{\mathcal{C}},\ddot{\boldsymbol{y}}+\widehat{\mathcal{C}}\right\}\in E_0 \text{ if and only if }
\left\{\boldsymbol{x}+\mathcal{C},\boldsymbol{y}+\mathcal{C}\right\}\in E_{\mathcal{C}}.$$
   Analogously, using the function $\varphi_1(\boldsymbol{x}+\mathcal{C})=\widetilde{\boldsymbol{x}}+\widehat{\mathcal{C}}$, it can be proved that $\Gamma(\mathcal{C})\cong\Gamma_1(\widehat{\mathcal{C}}).$
\end{proof}
\begin{definition}
    For a graph $G$ with vertex set $V$, another graph can be constructed as follows: let $V'$ be a set such that $V'\cap V=\emptyset$, $|V'|=|V|$, and $f:V\rightarrow V'$ a bijective function. For $a\in V$, the value of $f$ at $a$ is denoted by $a'$; that is, $a'=f(a)$. Consider the graph $DG$ with vertex set $V'\cup V$, whose edges are formed as follows: $\left\{a,b\right\}\in E_G$ if and only if $\left\{a,b'\right\},\left\{a',b\right\}\in E_{DG}$. The graph $DG$ is called the duplicate of $G$.   
\end{definition}
\begin{theorem}\label{teografoextendido}
    Let $\mathcal{C}$ be a binary linear code such that no codeword satisfies \eqref{eq:star}. Then 
    $$\Gamma(\widehat{\mathcal{C}})\cong D\Gamma(\mathcal{C})\cup H,$$
   where $H=(V(H),E(H))$ is the graph with
		\begin{align*}
		V(H)&=V_{\widehat{\mathcal{C}}}, \\ 
    E(H)&=\left\{\left\{\ddot{\boldsymbol{x}}+\widehat{\mathcal{C}}, \widetilde{\boldsymbol{x}}+\widehat{\mathcal{C}}\right\}:\ \boldsymbol{x}+\mathcal{C}\in V_{\mathcal{C}}\right\}.
		\end{align*}
\end{theorem}

\begin{proof}
    First, we show that the edges of the duplicate graph of $\widehat{\Gamma}(\mathcal{C})$ are given by $\left\{\ddot{\boldsymbol{x}}+\widehat{\mathcal{C}},\ddot{\boldsymbol{y}}+\widehat{\mathcal{C}}\right\}$ and $\left\{\widetilde{\boldsymbol{x}}+\widehat{\mathcal{C}},\widetilde{\boldsymbol{y}}+\widehat{\mathcal{C}}\right\}$, where $\left\{\boldsymbol{x}+\mathcal{C},\boldsymbol{y}+\mathcal{C}\right\}\in E_{\mathcal{C}}$.

    Indeed, consider the sets 
    \begin{align*}
        V&=\left\{\widehat{\boldsymbol{x}}+\widehat{\mathcal{C}}:\ \boldsymbol{x}+\mathcal{C}\in V_{\mathcal{C}} \right\}, \text{ and}\\
        V'&=\left\{\widehat{\boldsymbol{x}}+\boldsymbol{0}1+\widehat{\mathcal{C}}:\boldsymbol{x}+\mathcal{C}\in V_{\mathcal{C}}\right\}.
    \end{align*}
    Notice that $V\cap V'=\emptyset$ and $|V|=|V'|$. Moreover,
    $$\widehat{\boldsymbol{x}}+\boldsymbol{0}1=\left\{ \begin{array}{lcc} 
    \ddot{\boldsymbol{x}}+\boldsymbol{0}1=\boldsymbol{x}0+\boldsymbol{0}1=\boldsymbol{x}1=\widetilde{\boldsymbol{x}}, &\text{if} & \wt(\boldsymbol{x}) \text{ is even,}
    \\
    \widetilde{\boldsymbol{x}}+\boldsymbol{0}1=\boldsymbol{x}1+\boldsymbol{0}1=\boldsymbol{x}0=\ddot{\boldsymbol{x}}, &\text{if} & \wt(\boldsymbol{x}) \text{ is odd.}
    \end{array} \right.$$

    Thus, we define the function $f:V(\widehat{\Gamma}(\mathcal{C}))\rightarrow V'$, where
    $$f(\widehat{\boldsymbol{x}}+\widehat{\mathcal{C}})=\left\{ \begin{array}{lcc} 
\widetilde{\boldsymbol{x}}+\widehat{\mathcal{C}}, & & \text{if } \wt(\boldsymbol{x}) \text{ is even,}
    \\
    \ddot{\boldsymbol{x}}+\widehat{\mathcal{C}}, & & \text{if } \wt(\boldsymbol{x}) \text{ is odd.}
    \end{array} \right.$$
    Since $f$ is surjective and $|V|=|V'|$, it follows that $f$ is bijective.
    Consequently, from the definition of the duplicate of a graph, it follows that 
    $$\left\{ \widehat{\boldsymbol{x}}+\widehat{\mathcal{C}},\widehat{\boldsymbol{y}}+\widehat{\mathcal{C}} \right\}\in E(\widehat{\Gamma}(\mathcal{C})) \text{ if and only if } \left\{\widehat{\boldsymbol{x}}+\widehat{\mathcal{C}},f(\widehat{\boldsymbol{y}}+\widehat{\mathcal{C}}) \right\}, \left\{ f(\widehat{\boldsymbol{x}}+\widehat{\mathcal{C}}),\widehat{\boldsymbol{y}}+\widehat{\mathcal{C}} \right\}\in E(D\widehat{\Gamma}(\mathcal{C})).$$

    Additionally, consider $\left\{\boldsymbol{x}+\mathcal{C},\boldsymbol{y}+\mathcal{C}\right\} \in E_{\mathcal{C}}$. Note that if $\boldsymbol{x}+\mathcal{C}$ has even weight, then $\widehat{\boldsymbol{x}}+\widehat{\mathcal{C}}=\ddot{\boldsymbol{x}}+\widehat{\mathcal{C}}$ and $\boldsymbol{y}$ has odd weight. Thus, $\widehat{\boldsymbol{y}}+\widehat{\mathcal{C}}=\widetilde{\boldsymbol{y}}+\widehat{\mathcal{C}}$, which implies that $f(\widehat{\boldsymbol{y}}+\widehat{\mathcal{C}})=\ddot{\boldsymbol{y}}+\widehat{\mathcal{C}}$. 

    If $\boldsymbol{x}+\mathcal{C}$ has odd weight, then $\widehat{\boldsymbol{x}}+\widehat{\mathcal{C}}=\widetilde{\boldsymbol{x}}+\widehat{\mathcal{C}}$ and $\boldsymbol{y}$ has even weight. Thus, $\widehat{\boldsymbol{y}}+\widehat{\mathcal{C}}=\ddot{\boldsymbol{y}}+\widehat{\mathcal{C}}$, which implies that $f(\widehat{\boldsymbol{y}}+\widehat{\mathcal{C}})=\widetilde{\boldsymbol{y}}+\widehat{\mathcal{C}}$. Hence,
    $$ E_{D\widehat{\Gamma}(\mathcal{C})}= \left\{ \left\{ 
\ddot{\boldsymbol{x}}+\widehat{\mathcal{C}},\ddot{\boldsymbol{y}}+\widehat{\mathcal{C}}\right\},\left\{ 
\widetilde{\boldsymbol{x}}+\widehat{\mathcal{C}},\widetilde{\boldsymbol{y}}+\widehat{\mathcal{C}}\right\}:\left\{\boldsymbol{x}+\mathcal{C},\boldsymbol{y}+\mathcal{C}\right\} \in E_{\mathcal{C}}\right\}.$$
Finally, the result follows from Theorem \ref{teoExtendido}.
\end{proof}

Obviously, for an even code, there is no codeword satisfying \eqref{eq:star}; hence, the next  result follows from Theorem \ref{teografoextendido}.

    \begin{corollary}
    Let $\mathcal{C}$ be an \emph{even} binary linear code.  
		Then $$\Gamma(\widehat{\mathcal{C}})\cong D\Gamma(\mathcal{C})\cup H,$$
    where $H=(V(H),E(H))$, such that $V(H)=V_{\widehat{\mathcal{C}}}$  and  
     \[E(H)=\left\{\left\{\ddot{\boldsymbol{x}}+\widehat{\mathcal{C}}, \widetilde{\boldsymbol{x}}+\widehat{\mathcal{C}}\right\}:\ \boldsymbol{x}+\mathcal{C}\in V_{\mathcal{C}}\right\}.\]
\end{corollary}

Next, we present several examples applying Theorem \ref{teografoextendido} to construct the graphs of extended codes directly from the graph of the original code.

\begin{example} We use the following convention for the graph of the extended code (see Figure~\ref{fig:grafoextendido1}), given in the following examples: for $\mathcal{C}_i=\boldsymbol{x}_i+\mathcal{C}$, the 
two cosets $\ddot{\boldsymbol{x}}_i+\widehat{\mathcal{C}}$ and 
$\tilde{\boldsymbol{x}}_i+\widehat{\mathcal{C}}$ are represented by 
$\ddot{\mathcal{C}}_i$ and $\tilde{\mathcal{C}}_i$, respectively. Moreover, 
the black edges form the duplicate graph $D\widehat{\Gamma}(\mathcal{C})$, 
while the purple edges correspond to the graph $H$.  \label{ejemploextendido}
\begin{enumerate}[label=(\alph*)]

\item
\label{ejemplo2ext}
Let $\mathcal{C}=\{0000,1100,0011,1111 \}$ be a binary linear code; then, its extension is given by $\widehat{\mathcal{C}}=\{00000,11000,00110,11110\}$, which is a $[5,2,2]$-binary linear code. Thus, 
\begin{align*}
\mathfrak{cl}(\mathcal{C}) &= \{ 0000+\mathcal{C}, 1000+\mathcal{C}, 0010+\mathcal{C}, 1010+\mathcal{C} \}, \\
\mathfrak{cl}(\widehat{\mathcal{C}}) &= \Bigl\{ 00000+\widehat{\mathcal{C}}, 10000+\widehat{\mathcal{C}}, 00100+\widehat{\mathcal{C}}, 00001+\widehat{\mathcal{C}}, \\ 
&\phantom{{}= \Bigl\{} 10100+\widehat{\mathcal{C}}, 10001+\widehat{\mathcal{C}}, 00101+\widehat{\mathcal{C}}, 10101+\widehat{\mathcal{C}} \Bigr\}.
\end{align*}
 \begin{figure}[H]
    	\begin{subfigure}[b]{0.49\textwidth}
    	\centering
    	\includegraphics[scale=0.7]{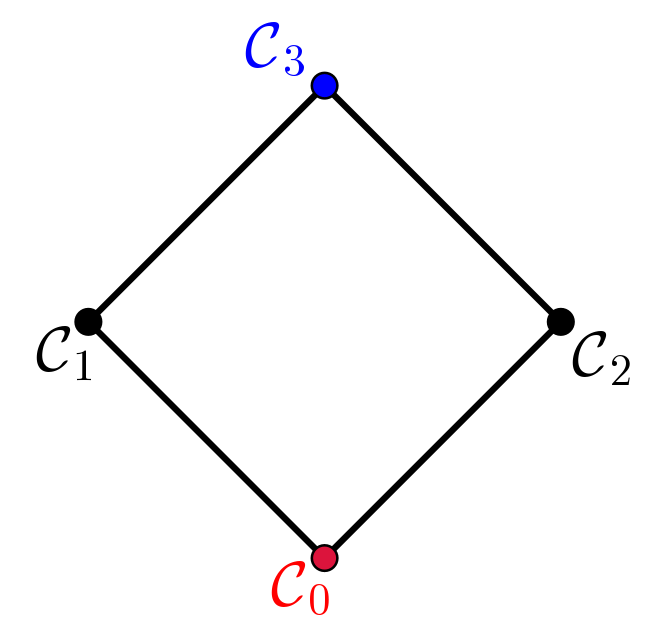} 
    	\caption{Graph of $\Gamma(\mathcal{C})$.}
    	\label{fig:grafoor1xCap4}
    	\end{subfigure}
    	\begin{subfigure}[b]{0.49\textwidth}
    		\centering
    		\includegraphics[scale=0.7]{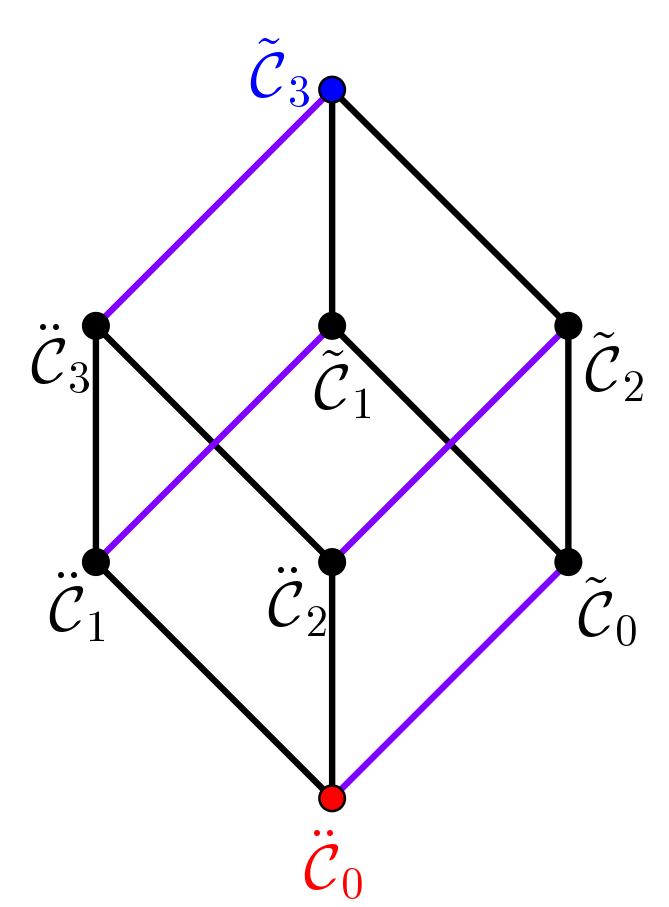}
    		\caption{Graph of $\Gamma(\widehat{\mathcal{C}})$.}
    		\label{fig:grafoextendido1}
    	\end{subfigure}
        \caption{Graph of a extended code.}
        \label{fig:ejemploextendido1}
    \end{figure}

\item\label{ejemplo1ext}
Let $\mathcal{C}=\{000, 100\}$ be a $[3,1,1]$-binary linear code. Then, its extension  is $\widehat{\mathcal{C}}=\{0000,1001\}$ is a $[4,1,2]$-binary linear code.   In this code, there exist codewords that satisfy condition~\eqref{eq:star}. However, the characterization given in Theorem~\ref{teografoextendido} still holds for the graph of the extended code; that is, no new edges are generated beyond those 
provided in that result, see Figure \ref{fig:ejemploextendido1}. Indeed, observe  that the cosets leaders of $\mathcal{C}$ and $\widehat{\mathcal{C}}$ are: $000,010,001,011$ and $0000,0100,0010,0001,0110,0101$,  $0011,0111$, respectively. 

\item
 \label{ejemplo3ext}
    For the $[5,2,2]$-binary linear code $\mathcal{C}=\{00000,00111,11000,11111 \}$, we have that $\widehat{\mathcal{C}}=\{000000,001111,110000,111111\}$ is a $[6,2,2]$-code. Note that the codeword $00111$ satisfies $\wt(\boldsymbol{c})=2\wt(\boldsymbol{x}\cap\boldsymbol{c})+1$ for $\boldsymbol{x}$ being a coset leader of: $00100+\mathcal{C}$, $00010+\mathcal{C}$, $00001+\mathcal{C}$, $10100+\mathcal{C}$, $10010+\mathcal{C}$, and $10001+\mathcal{C}$; 
whereas the codeword $11111$ satisfies condition~\eqref{eq:star} with the cosets $10100+\mathcal{C}$, $10010+\mathcal{C}$, and $10001+\mathcal{C}$.  However, in cosets of $\widehat{\mathcal{C}}$, new coset leaders appear that are not 
of the form $\ddot{\boldsymbol{x}}$ or $\tilde{\boldsymbol{x}}$, where $\boldsymbol{x}$ is a coset leader of $\mathcal{C}$. 		As in the previous example, the black edges represent the duplicate graph $D\widehat{\Gamma}(\mathcal{C})$ and the purple edges correspond to the graph $H$. Additionally, in this case, new edges arise, which have been colored in pink. These represent the new relations obtained from the new coset leaders  that appeared in the standard array of the extended code, see Figure \ref{fig:subgrafo1}.

\begin{figure}[H]
    	\begin{subfigure}{\textwidth}
    	\centering
    	\includegraphics[scale=0.7]{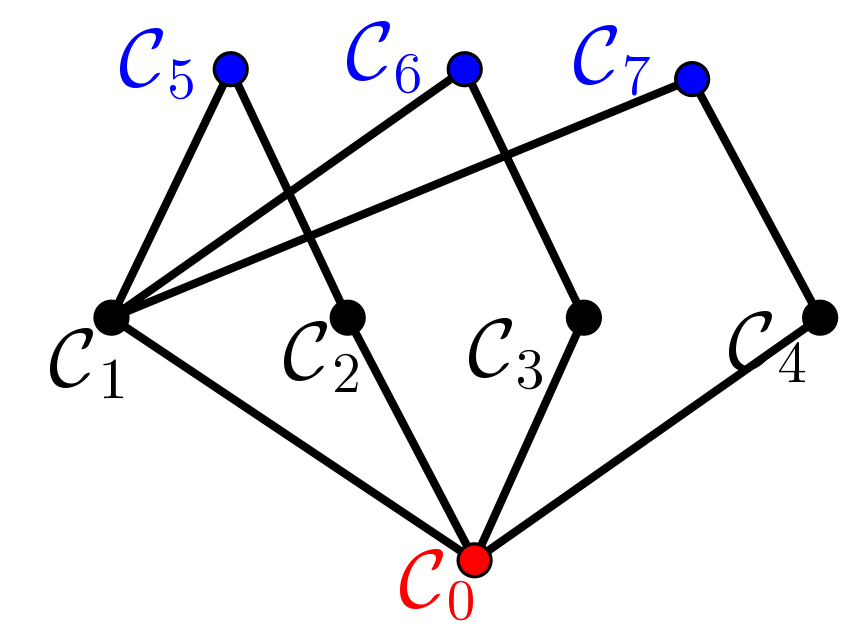} 
    	\caption{Graph $\Gamma(\mathcal{C})$.}
    	\label{fig:grafoor12}
    	\end{subfigure}
    
    	\begin{subfigure}{\textwidth}
    		\centering
    		\includegraphics[scale=0.8]{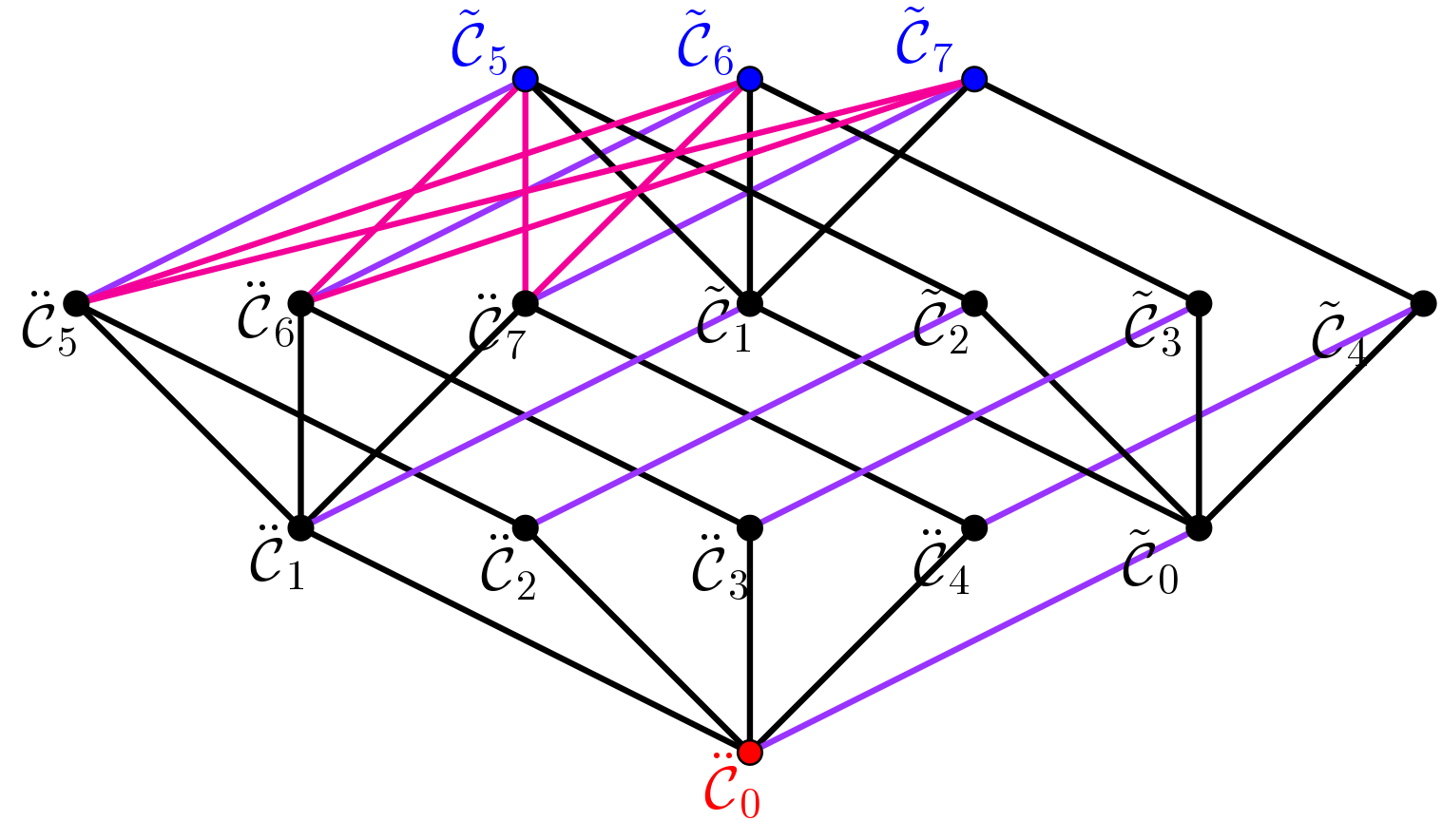}
    		\caption{Graph $\Gamma(\widehat{\mathcal{C}})$.}
    		\label{fig:subgrafo1}
    	\end{subfigure}
        \caption{Construction of $\Gamma(\widehat{\mathcal{C}})$ from $\Gamma(\mathcal{C})$. }
\label{fig:extendido3}
    \end{figure}

    \end{enumerate}
\end{example}

\subsection{Direct sum}
 \begin{definition}[Direct sum]
    Let $\mathcal{C}$ be an $[n_1,k_1,d_1]$-binary linear code and $\mathcal{D}$ be an $[n_2,k_2,d_2]$-binary linear code. The \emph{direct sum} of $\mathcal{C}$ and $\mathcal{D}$ is given by
    $$\mathcal{C} \oplus \mathcal{D}=\{\boldsymbol{c}\boldsymbol{d}:\ \boldsymbol{c}\in \mathcal{C},\ \boldsymbol{d}\in \mathcal{D}\},$$
where $\boldsymbol{c}\boldsymbol{d}$ means the concatenation of the vectors.
\end{definition}

\begin{example}\label{ejemplosumadirecta}
 \quad
\begin{enumerate}

		    \item \label{sumadirecta1} Let $\mathcal{C}=\{00,01\}$ and $\mathcal{D}=\{000,010,100,110\}$ be two binary linear codes. The direct sum of these two codes is a binary linear $[5,3,1]$-code given by $$\mathcal{C}\oplus\mathcal{D}=\left\{00000,00010,00100,00110,01000,01010,01100,01110\right\}.$$
				\item \label{sumadirecta3} Consider the binary linear codes $\mathcal{C}=\{00, 01\}$ and $\mathcal{D}=\{000,111\}$. Then $\mathcal{C}\oplus\mathcal{D}$ is a $[5,2,1]$-binary linear code such that 
    $$\mathcal{C}\oplus\mathcal{D}=\left\{00000,00111,01000,01111\right\}.$$

    \item \label{sumadirecta2} Let $\mathcal{C}=\{000,110,011,101\}$ and $\mathcal{D}=\{0000,1000,0001,1001\}$ be two binary linear codes. The direct sum of these two codes is a binary linear $[7,4,1]$-code given by 
    \begin{align*}
    \mathcal{C}\oplus\mathcal{D}&=\left\{0000000,0001000,0000001,0001001,1100000,1101000,1100001,1101001,\right.\\
    &\left.\quad 1010000,1011000,1010001,1011001,0110000,0111000,0110001,0111001\right\}.
    \end{align*}
\end{enumerate}
  \end{example} 
	The following theorem serves to characterize the parameters of the direct sum code associated with two given codes.
\begin{theorem}[{\cite[Theorem 6.1.5]{Ling2004}}]\label{teosumadirecta}
    If $\mathcal{C}$ is an $[n_1,k_1,d_1]$-binary linear code and $\mathcal{D}$ is an $[n_2,k_2,d_2]$-binary linear code, then $\mathcal{C} \oplus \mathcal{D}$ is an $\left[n_1+n_2,k_1+k_2,\min\{d_1,d_2\}\right]$-binary linear code.
\end{theorem}
		
Consequently, from Theorem \ref{maintheoremarticle1} and Theorem \ref{teosumadirecta}, we get the next characterization of the order of $\Gamma(\mathcal{C}\oplus\mathcal{D})$.  
\begin{corollary}
\label{corcardinalsumadirecta}
    Let $\mathcal{C}$ be an $[n_1,k_1,d_1]$-binary linear code and $\mathcal{D}$ be an $[n_2,k_2,d_2]$-binary linear code. Then, $|V_{\mathcal{C}\oplus\mathcal{D}}|=2^{(n_1+n_2)-(k_1+k_2)}$. 
\end{corollary}
In the next step, we study the structure of the cosets leaders of $\mathcal{C}\oplus\mathcal{D}$.
\begin{proposition}\label{teor:lideres_suma_directa}
    Let $\mathcal{C}$ be an $[n_1,k_1,d_1]$-binary linear code and $\mathcal{D}$ be an $[n_2,k_2,d_2]$-binary linear code. Take $\boldsymbol{x}_1,\boldsymbol{y}_1\in \mathbb{F}_2^{n_1}$ and $\boldsymbol{x}_2,\boldsymbol{y}_2\in \mathbb{F}_2^{n_2}$. If $\boldsymbol{x}_1\not\in\boldsymbol{y}_1+\mathcal{C}$ or $\boldsymbol{x}_2\not\in\boldsymbol{y}_2+\mathcal{D}$, then $\boldsymbol{x}_1\boldsymbol{x}_2\not\in \boldsymbol{y}_1\boldsymbol{y}_2+\mathcal{C}\oplus\mathcal{D}$. Moreover, if $\boldsymbol{x}_1$ and $\boldsymbol{x}_2$ are coset leaders of $\boldsymbol{x}_1+\mathcal{C}$ and $\boldsymbol{x}_2+\mathcal{D}$ respectively, then $\boldsymbol{x}_1\boldsymbol{x}_2$ is a coset leader of $\boldsymbol{x}_1\boldsymbol{x}_2+\mathcal{C}\oplus\mathcal{D}$. 
\end{proposition}
\begin{proof}
    By \cite[Section 1.5.4]{Huffman2003}, a parity-check matrix of $\mathcal{C}\oplus\mathcal{D}$ is given by 
    \begin{equation*}
        \left(\begin{array}{cc}
  H_1& O \\
  O & H_2
\end{array}\right),
    \end{equation*}
where $O$ stands for the zero matrix and $H_1$ and $H_2$ are parity-check matrices of $\mathcal{C}$ and $\mathcal{D}$, respectively. Note that
the two zero matrices have different sizes.
 
In this manner, the syndrome of $\boldsymbol{x}_1\boldsymbol{x}_2\in\mathbb{F}_2^{n_1+n_2}$ is given by
    \begin{align*}
        \Syn(\boldsymbol{x}_1\boldsymbol{x}_2)&=\boldsymbol{x}_1\boldsymbol{x}_2 \left(
  \begin{array}{cc}
  H_1& O \\
  O & H_2
\end{array} \right)^T\\
&= \boldsymbol{x}_1\boldsymbol{x}_2 \left(
  \begin{array}{cc}
  H_1^T& O \\
  O & H_2^T
\end{array} \right)\\
&=(\Syn(\boldsymbol{x}_1),\Syn(\boldsymbol{x}_2)).
    \end{align*}
Similarly, $\Syn(\boldsymbol{y}_1\boldsymbol{y}_2)= (\Syn(\boldsymbol{y}_1),\Syn(\boldsymbol{y}_2))$. However, by hypothesis, $\boldsymbol{x}_1\not \in \boldsymbol{y}_1+\mathcal{C} $ or $\boldsymbol{x}_2\not \in \boldsymbol{y}_2+\mathcal{D}$; that is, $\Syn( \boldsymbol{x}_1)\neq \Syn(\boldsymbol{y}_1)$ or $\Syn( \boldsymbol{x}_2)\neq \Syn(\boldsymbol{y}_2)$. Therefore, $\Syn( \boldsymbol{x}_1\boldsymbol{x}_2)\neq \Syn (\boldsymbol{y}_1\boldsymbol{y}_2)$, which implies that $\boldsymbol{x}_1\boldsymbol{x}_2\not\in \boldsymbol{y}_1\boldsymbol{y}_2+\mathcal{C}\oplus\mathcal{D}$.

Now, note that $\boldsymbol{x}_1\boldsymbol{x}_2$ is a coset leader of the coset $\boldsymbol{x}_1\boldsymbol{x}_2+\mathcal{C}\oplus\mathcal{D}$ since, by hypothesis, 
\begin{align}\label{eq:des_peso_suma_directa}
\begin{split}
    \wt(\boldsymbol{x}_1)&\leq \wt(\boldsymbol{x}_1+\boldsymbol{c}_1), \quad \forall \boldsymbol{c}_1\in \mathcal{C}\\
    \wt(\boldsymbol{x}_2)&\leq \wt(\boldsymbol{x}_2+\boldsymbol{c}_2), \quad \forall \boldsymbol{c}_2\in \mathcal{D}, 
\end{split}
\end{align}
when $\boldsymbol{x}_1+\boldsymbol{c}_1\notin \mathfrak{cl}(\mathcal{C})$ and $\boldsymbol{x}_2+\boldsymbol{c}_2\notin \mathfrak{cl}(\mathcal{D})$.

By inequalities \eqref{eq:des_peso_suma_directa}, and given that $\wt(\boldsymbol{z}_1\boldsymbol{z}_2)=\wt(\boldsymbol{z}_1)+\wt(\boldsymbol{z}_2)$ for any $\boldsymbol{z}_1\in\mathbb{F}_2^{n_1}$ and $\boldsymbol{z}_2\in\mathbb{F}_2^{n_2}$, it follows that 
$$\wt(\boldsymbol{x}_1\boldsymbol{x}_2)\leq \wt((\boldsymbol{x}_1+\boldsymbol{c}_1)(\boldsymbol{x}_2+\boldsymbol{c}_2)).$$ 
Thus, $\boldsymbol{x}_1\boldsymbol{x}_2$ is a coset leader of the coset $\boldsymbol{x}_1\boldsymbol{x}_2+\mathcal{C}\oplus\mathcal{D}$. 
\end{proof}

		\begin{proposition}
The cosets of $\mathcal{C}\oplus\mathcal{D}$ are given by
    $$V_{\mathcal{C}\oplus\mathcal{D}}=\Bigl\{\boldsymbol{x}_1\boldsymbol{x}_2+\mathcal{C}\oplus\mathcal{D}:\boldsymbol{x}_1+\mathcal{C}\in V_{\mathcal{C}} \text{ and }\boldsymbol{x}_2+\mathcal{D}\in V_{\mathcal{D}}\Bigr\}.$$
\end{proposition}

\begin{proof}
    By Proposition \ref{teor:lideres_suma_directa}, each $\boldsymbol{x}_1\boldsymbol{x}_2+\mathcal{C}\oplus\mathcal{D}$ generates a coset in $\mathcal{C}\oplus\mathcal{D}$ for each $\boldsymbol{x}_1\in \mathbb{F}_2^{n_1}$ and $\boldsymbol{x}_2\in \mathbb{F}_2^{n_2}$. Since it is known that $|\mathfrak{cl}(\mathcal{C})|=2^{n_1-k_1}$ and $|\mathfrak{cl}(\mathcal{D})|=2^{n_2-k_2}$, then $|\mathfrak{cl}(\mathcal{C}\oplus\mathcal{D})|=2^{n_1-k_1}2^{n_2-k_2}=2^{(n_1+n_2)-(k_1+k_2)}$ cosets. Thus, the result follows from Theorem \ref{maintheoremarticle1} and Theorem \ref{teosumadirecta}. 
\end{proof}
The following definition will be used to characterize the graph of the direct sum code of two given codes.

\begin{definition}
The cartesian product of two simple graphs $H$ and $K$ is the graph $G=H\times K$ with $V(G)=V(H)\times V(K)$, in which two vertices $\{h,k\}$ and $\{h',k'\}$ are adjacent if and only if
\begin{enumerate}
    \item $k=k'$ and $\{h,h'\}\in E(H)$, or
    \item $h=h'$ and $\{k,k'\}\in E(K)$.
\end{enumerate}
\end{definition}

 \begin{theorem}\label{teor:grafo_sumadirecta}
    Let $\mathcal{C}$ be an $[n_1,k_1,d_1]$-binary linear code, $\mathcal{D}$ be an $[n_2,k_2,d_2]$-binary linear code, and $\Gamma(\mathcal{C})$ and $\Gamma(\mathcal{D})$ its associated Hasse diagram, respectively. Then 
$$\Gamma(\mathcal{C}\oplus\mathcal{D})\cong\Gamma(\mathcal{C})\times\Gamma(\mathcal{D}),$$
where
    \begin{align*}
    V_{\mathcal{C}\oplus\mathcal{D}}&=\Bigl\{\boldsymbol{x}_1\boldsymbol{x}_2+\mathcal{C}\oplus\mathcal{D}:\boldsymbol{x}_1+\mathcal{C}\in V_{\mathcal{C}} \text{ and } \boldsymbol{x}_2+\mathcal{D}\in V_{\mathcal{D}} \Bigr\},\\
    E_{\mathcal{C}\oplus\mathcal{D}}&= \Bigl\{ 
    \{ \boldsymbol{x}_1\boldsymbol{x}_2+\mathcal{C}\oplus\mathcal{D}, \boldsymbol{y}_1\boldsymbol{x}_2+\mathcal{C}\oplus\mathcal{D}
    \}: \{\boldsymbol{x}_1+\mathcal{C},\boldsymbol{y}_1+\mathcal{C}\}\in E_{\mathcal{C}}\Bigr\} \\
    &\qquad\bigcup\Bigl\{ 
    \{ \boldsymbol{x}_1\boldsymbol{x}_2+\mathcal{C}\oplus\mathcal{D}, \boldsymbol{x}_1\boldsymbol{y}_2+\mathcal{C}\oplus\mathcal{D}
    \}: \{\boldsymbol{x}_2+\mathcal{D},\boldsymbol{y}_2+\mathcal{D}\}\in E_{\mathcal{D}} \Bigr\}.  
    \end{align*}
 In particular, $|V_{\mathcal{C}\oplus\mathcal{D}}|=|V_{\mathcal{C}}||V_{\mathcal{D}}|$ and $|E_{\mathcal{C}\oplus\mathcal{D}}|=|V_{\mathcal{C}}||E_{\mathcal{D}}|+|E_{\mathcal{C}}||V_{\mathcal{D}}|$.
 \end{theorem}		

\begin{proof}
It only remains to analyze the cases in which there is an edge in the graph of the direct sum. We study the next cases: 
\begin{enumerate}
\item Suposse that $\boldsymbol{x}_1\nprec \boldsymbol{y}_1$ (or $\boldsymbol{x}_2\nprec \boldsymbol{y}_2$). This means that $\supp(\boldsymbol{x}_1)\not\subset\supp(\boldsymbol{y}_1)$; therefore, $\supp(\boldsymbol{x}_1\boldsymbol{x}_2)\not\subset\supp(\boldsymbol{y}_1\boldsymbol{y}_2)$, which implies that $\boldsymbol{x}_1\boldsymbol{x}_2\nprec\boldsymbol{y}_1\boldsymbol{y}_2$. Similarly, the same result is obtained if $\boldsymbol{x}_2\nprec \boldsymbol{y}_2$.

\item Assume that $\boldsymbol{x}_1\preceq \boldsymbol{y}_1$ and $\boldsymbol{x}_2\preceq \boldsymbol{y}_2$. Then, $\wt(\boldsymbol{x}_1)\leq \wt(\boldsymbol{y}_1)$ and $\wt(\boldsymbol{x}_2)\leq \wt(\boldsymbol{y}_2)$, thus
    $\wt(\boldsymbol{x}_1\boldsymbol{x}_2)\leq \wt(\boldsymbol{y}_1\boldsymbol{y}_2)$. We must consider the following subcases:
\begin{enumerate}

\item $\{\boldsymbol{x}_1+\mathcal{C},\boldsymbol{y}_1+\mathcal{C}\}\in E_{\mathcal{C}}$. Then, $\wt(\boldsymbol{x}_1)= \wt(\boldsymbol{y}_1)-1$ and $\wt(\boldsymbol{x}_1\boldsymbol{x}_2)=\wt(\boldsymbol{y}_1\boldsymbol{x}_2)-1$. Furthermore, since $\supp(\boldsymbol{x}_1)\subset \supp(\boldsymbol{y}_1)$, it follows that $\supp(\boldsymbol{x}_1\boldsymbol{x}_2)\subset \supp(\boldsymbol{y}_1\boldsymbol{x}_2)$. That is, $\boldsymbol{x}_1\boldsymbol{x}_2\prec\boldsymbol{y}_1\boldsymbol{x}_2$. Therefore, $\{\boldsymbol{x}_1\boldsymbol{x}_2+\mathcal{C}\oplus\mathcal{D},\boldsymbol{y}_1\boldsymbol{x}_2+\mathcal{C}\oplus\mathcal{D}\}\in E_{\mathcal{C}\oplus\mathcal{D}}$.

\item $\{\boldsymbol{x}_2+\mathcal{D},\boldsymbol{y}_2+\mathcal{D}\}\in E_{\mathcal{D}}$. Similarly to the previous subcase, we have
$\wt(\boldsymbol{x}_1\boldsymbol{x}_2)=\wt(\boldsymbol{x}_1\boldsymbol{y}_2)-1$ and $\supp(\boldsymbol{x}_1\boldsymbol{x}_2)\subset \supp(\boldsymbol{x}_1\boldsymbol{y}_2)$. Thus, $\{\boldsymbol{x}_1\boldsymbol{x}_2+\mathcal{C}\oplus\mathcal{D},\boldsymbol{x}_1\boldsymbol{y}_2+\mathcal{C}\oplus\mathcal{D}\}\in E_{\mathcal{C}\oplus\mathcal{D}}.$

\item $\{\boldsymbol{x}_1+\mathcal{C},\boldsymbol{y}_1+\mathcal{C}\}\in E_{\mathcal{C}}$ and $\{\boldsymbol{x}_2+\mathcal{D},\boldsymbol{y}_2+\mathcal{D}\}\in E_{\mathcal{D}}$. Then, $\wt(\boldsymbol{x}_1)=\wt(\boldsymbol{y}_1)-1$ and $\wt(\boldsymbol{x}_2)=\wt(\boldsymbol{y}_2)-1$. 
Hence, $\wt(\boldsymbol{x}_1\boldsymbol{x}_2)=\wt(\boldsymbol{y}_1\boldsymbol{y}_2)-2$. Therefore, $\boldsymbol{x}_1\boldsymbol{x}_2\not\prec\boldsymbol{y}_1\boldsymbol{y}_2$.\qedhere
\end{enumerate}
\end{enumerate}
\end{proof}

\begin{example}\quad
\label{ejemplosGrafoumadirecta}
Below, figures \ref{fig1:suma_directa}, \ref{fig3:suma_directa} and  \ref{fig2:suma_directa} show the graphs of the codes from Example \ref{ejemplosumadirecta}, along with the respective graph of the direct sum between the two given codes. We give the construction of the direct sum graph of the codes is carried out by applying Theorem \ref{teor:grafo_sumadirecta}, where $\mathcal{C}_i\mathcal{D}_j$ denotes the coset $\boldsymbol{x}_i\boldsymbol{y}_j+\mathcal{C}\oplus \mathcal{D}$. As a result of this process, the third graph shown in each case is obtained.

\begin{figure}[H]
    \begin{subfigure}[b]{0.2\textwidth}
    \centering
    \includegraphics[scale=0.8]{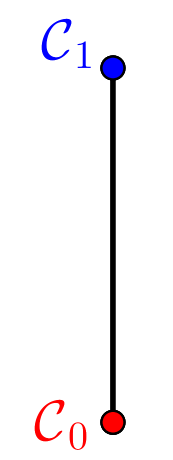} 
    \caption{Graph $\Gamma(\mathcal{C}_1)$.}
    \label{fig:grafoor16}
    \end{subfigure}
    \begin{subfigure}[b]{0.3\textwidth}
    \centering
    \includegraphics[scale=0.8]{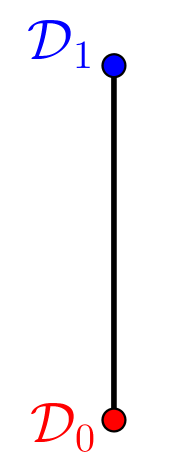} 
    \caption{Graph $\Gamma(\mathcal{C}_2)$.}
    \label{fig:grafoor17}
    \end{subfigure}
    \begin{subfigure}[b]{0.4\textwidth}
        \centering
        \includegraphics[scale=0.75]{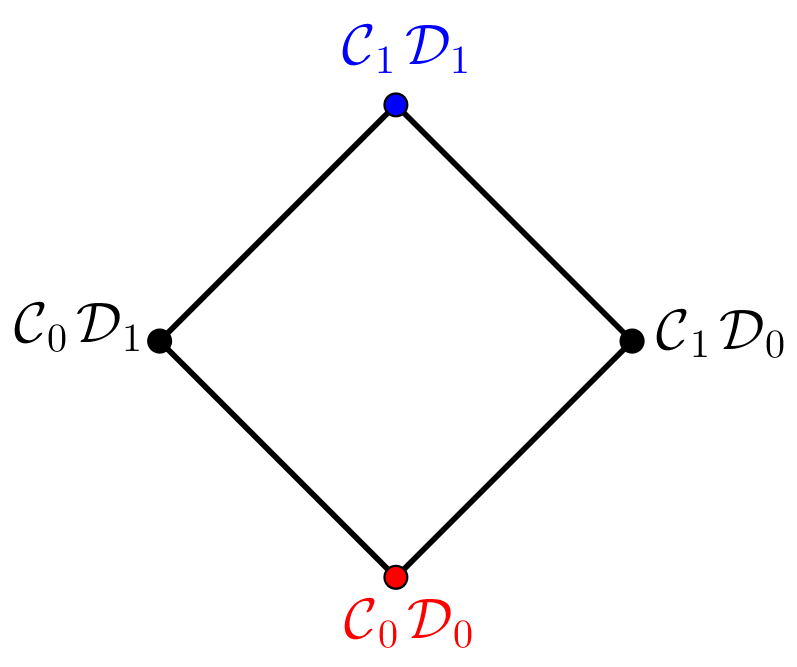}
        \caption{Graph $\Gamma(\mathcal{C}_1\oplus\mathcal{C}_2)$.}
        \label{fig:subgrafo5}
    \end{subfigure}
    \caption{Graphs of the codes from Example \ref{ejemplosumadirecta} \ref{sumadirecta1}.}
    \label{fig1:suma_directa}
    \end{figure}

\begin{figure}[H]
   \begin{subfigure}[b]{0.2\textwidth}
    \centering
    \includegraphics[scale=0.8]{ejemplosm1v2.png} 
    \caption{Graph $\Gamma(\mathcal{C})$.}
    \label{fig:grafoor18}
    \end{subfigure}
    \begin{subfigure}[b]{0.3\textwidth}
    \centering
    \includegraphics[scale=0.75]{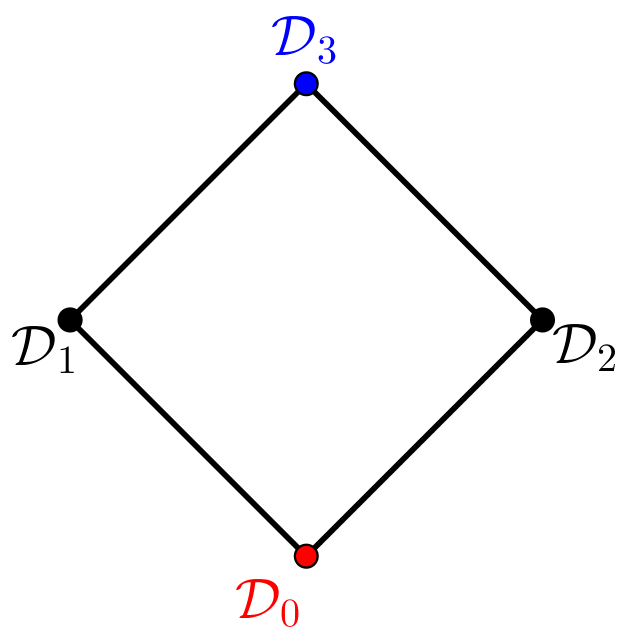} 
    \caption{Graph $\Gamma(\mathcal{D})$.}
    \label{fig:grafoor1}
    \end{subfigure}
    \begin{subfigure}[b]{0.4\textwidth}
        \centering
        \includegraphics[scale=0.7]{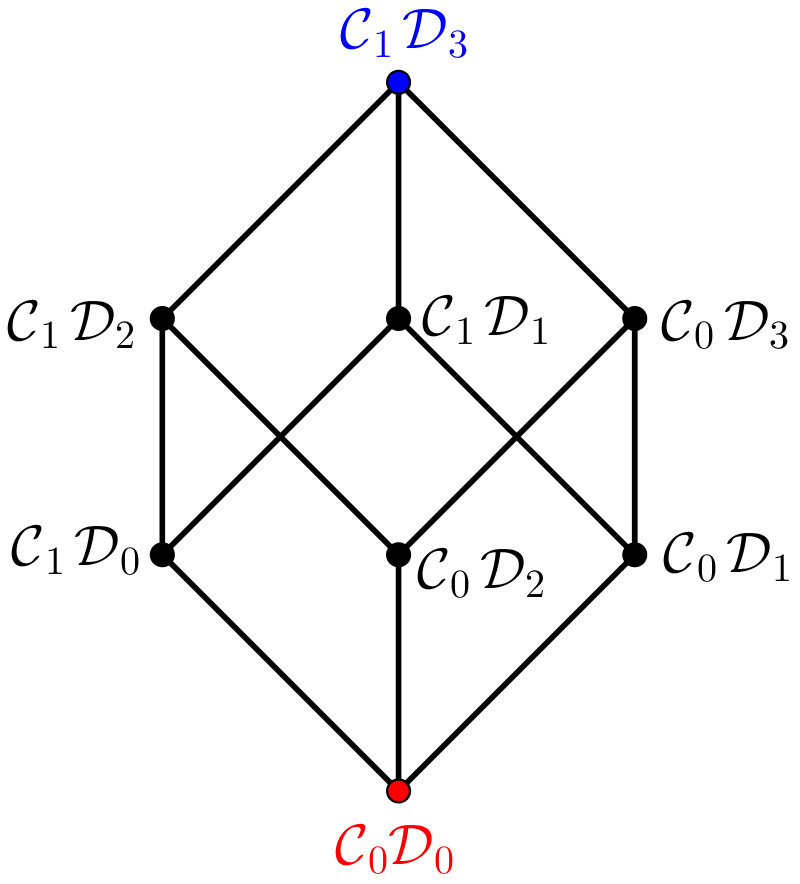}
        \caption{Graph $\Gamma(\mathcal{C}\oplus\mathcal{D})$.}
        \label{fig:subgrafo6}
    \end{subfigure}
    \caption{Graphs of the codes from Example \ref{ejemplosumadirecta} \ref{sumadirecta2}.}
    \label{fig3:suma_directa}
    \end{figure}

 \begin{figure}[H]
    	\begin{subfigure}[b]{0.2\textwidth}
    	\centering
    	\includegraphics[scale=0.7]{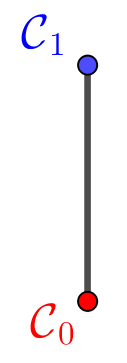} 
    	\caption{Graph $\Gamma(\mathcal{C})$.}
    	\label{fig:grafoor18}
    	\end{subfigure}
            \begin{subfigure}[b]{0.3\textwidth}
    	\centering
    	\includegraphics[scale=0.7]{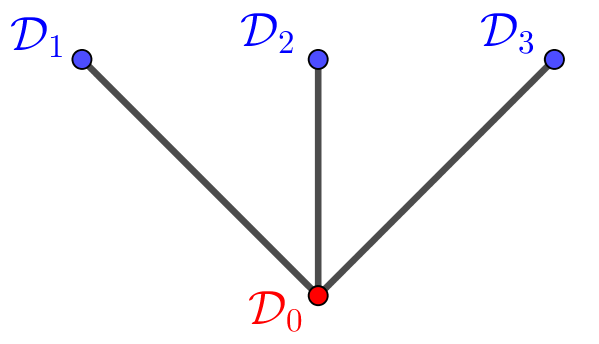} 
    	\caption{Graph  $\Gamma(\mathcal{D})$.}
    	\label{fig:grafoor1}
    	\end{subfigure}
    	\begin{subfigure}[b]{0.4\textwidth}
    		\centering
    		\includegraphics[scale=0.6]{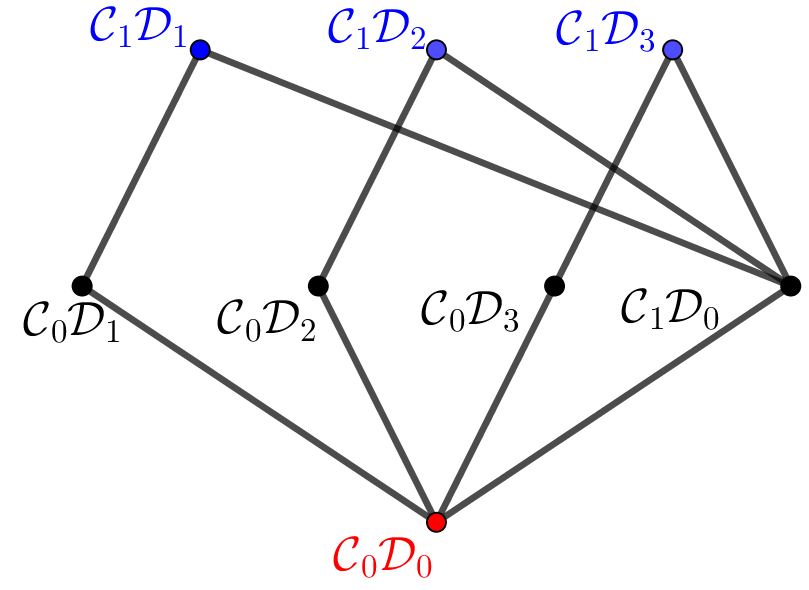}
    		\caption{Graph $\Gamma(\mathcal{C}\oplus\mathcal{D})$.}
    		\label{fig:subgrafo6}
    	\end{subfigure}
        \caption{Graph of the direct sum.}
        \label{fig2:suma_directa}
    \end{figure}

\end{example}

\bibliographystyle{plain}

\end{document}